\documentclass{llncs}
\usepackage{stmaryrd}
\usepackage{amsmath}
\usepackage{amssymb}
\usepackage{autobreak}
\usepackage{multirow}
\usepackage[table,xcdraw]{xcolor}
\usepackage{bussproofs}
\usepackage{pgfplots}
\usepackage{color}
\usepackage[utf8]{inputenc}
\usetikzlibrary{quantikz2}
\usepackage{xspace}
\usepackage{amsfonts}
\usepackage{todonotes}
\usepackage{url}
\usepackage{caption}
\usepackage{cleveref}
\usepackage{subcaption}
\usepackage{xcolor}

\usepackage{mathtools}
\usetikzlibrary{shapes}
\usepackage{wrapfig}
\usepackage{mathabx}
\usepackage{nicefrac}
\usepackage{mdframed}

\definecolor{mygrey}{rgb}{0.9, 0.95, 0.95}
\definecolor {iron}{RGB}{72, 73, 75}
\definecolor{yell}{RGB}{246,187,0}
\definecolor{yell}{RGB}{246,187,0}
\definecolor{blu}{RGB}{0,20,100}
\definecolor{re}{RGB}{148,6,39}
\definecolor{purp}{RGB}{74,13,70}
\definecolor{npurp}{RGB}{180,242,185}
\definecolor{gree}{RGB}{0,100,0}
\definecolor{rgree}{RGB}{74,53,20}
\definecolor{orange}{RGB}{255,178,102}

\newcommand{\CNOT}[0]{\ensuremath{\mathsf{{CNOT}}\xspace}}

\newcommand{\Hg}[0]{\ensuremath{\mathsf{H}\xspace}}

\newcommand{\Rx}{\ensuremath{\mathsf{Rx}\xspace}}
\newcommand{\Ry}{\ensuremath{\mathsf{Ry}\xspace}}
\newcommand{\Rz}{\ensuremath{\mathsf{Rz}\xspace}}
\newcommand{\CRz}{\ensuremath{\mathsf{CRz}\xspace}}

\newcommand{\pasum}[4]{\ensuremath{\left\langle #1,#2,#3\right\rangle_{#4}}\xspace}

\newcommand{\Mon}{\ensuremath{\texttt{Mon}}\xspace}

\newcommand{\boolcond}{\ensuremath{\texttt{bc}}\xspace}

\newcommand{\had}[0]{\ensuremath{\mathsf{{H}}}\xspace}

\newcommand{\zz}[0]{\ensuremath{\mathsf{{Z}}\xspace}}

\newcommand{\ps}[0]{PS\xspace}

\newcommand{\var}[1]{\ensuremath{\mathit{var(#1)}}\xspace}
\renewcommand{\vec}[1]{\ensuremath{\mathbf{#1}}\xspace}

\newcommand{\kett}[1]{\ensuremath{|#1 \rangle}}

\renewcommand{\ket}[2]{\ensuremath{|#1 \rangle_{#2}}}

\newcommand{\toolname}[0]{\textsc{QuPRS}\xspace}
\newcommand{\symp}[0]{\textsc{Symengine}\xspace \cite{symengine}\xspace}
\newcommand{\symf}[0]{\textsc{Sympy}\xspace\cite{meurer2017sympy}\xspace}
\newcommand{\feymann}[0]{\textsc{Feymann}\xspace}
\newcommand{\quokka}[0]{\textsc{Quokka\#}\xspace}

\newcommand{\qiskit}[0]{\textsc{qiskit}\xspace\cite{qiskit2024} \xspace}

\newcommand{\wmc}[0]{\mathit{WMC}\xspace}

\providecommand{\vect}[1]{\ensuremath{( \begin{matrix} #1 \end{matrix} )}}

\newcommand{\pyzx}[0]{\textsc{PyZX}\xspace}
\newcommand{\tdd}[0]{\textsc{TDD}\xspace}

\newcommand{\sliqec}[0]{\textsc{SliQEC}\xspace}

\newcommand{\qcec}[0]{\textsc{Qcec}\xspace}

\newcommand\smat[1]{\ensuremath{\left[\begin{smallmatrix}#1\end{smallmatrix}\right]}}
\usepackage[table]{xcolor} 

\definecolor{mygreen}{HTML}{38FFF8}
\definecolor{myyellow}{HTML}{FFFE65}

\title{Equivalence Checking of Quantum Circuits via Path-Sum and Weighted Model Counting}

\author{
Wei-Jia Huang\inst{1}\thanks{wei.jia.huang.physics@gmail.com. The work was conducted while he was a researcher at Foxconn Research.} \and
Christophe Chareton\inst{2}\thanks{christophe.chareton@cea.fr} \and
Yu-Fang Chen\inst{3}\thanks{yfc@iis.sinica.edu.tw} \and
Kai-Min Chung\inst{3} \and
Min-Hsiu Hsieh\inst{4} \and
Alfons Laarman\inst{5} \and
Jingyi Mei\thanks{j.mei@liacs.leidenuniv.nl}\inst{5}
}

\institute{
Quanta Computer Inc., Taoyuan, Taiwan \and
CEA List, Paris-Saclay, France \and
Institute of Information Science, Academia Sinica, Taipei, Taiwan \and
Foxconn Research, Taipei, Taiwan \and
LIACS, Leiden University, Leiden, The Netherlands
}

\begin{document}


\maketitle
\begin{abstract}
    Equivalence checking of quantum circuits is a central verification task in quantum computing, ensuring the correctness of circuit optimizations, hardware mappings, and compilation pipelines. Among the primary symbolic methods for this purpose, the \emph{path-sum formalism} provides a compact representation with powerful reduction rules that yield a canonical form for 
    the classically simulable Clifford fragment, 
    but \emph{confluence} fails beyond the Clifford fragment. We introduce a new \emph{weighted model counting} (WMC) encoding for path-sums and combine it with the existing path-sum reductions to obtain a verifier that is both complete and efficient. Our method applies reductions whenever possible and invokes the WMC-based decision procedure on the residual path-sum, yielding a complete semantic check up to a global phase. We implement the approach and evaluate it on standard benchmarks. Results show that the hybrid method outperforms either component in isolation and competes with state-of-the-art~tools.    
\end{abstract}

\setcounter{tocdepth}{4}
\section{Introduction}  
Quantum computation stores and processes information in systems governed by quantum mechanics~\cite{nielsen2002quantum}. Quantum devices are rapidly maturing~\cite{arute2019quantum,acharya2024quantum}. Unlike classical data, quantum data cannot be cloned, and evolution is described by unitary operators. 
In practice, computations are specified by sequences of such unitaries, formalized as \emph{quantum circuits}. 
Quantum circuits are typically executed on a coprocessor architecture: a classical host emits instructions that initialize a quantum register and apply these gates.  

As these circuits are compiled for execution, they undergo successive transformations—optimizations, decompositions, and hardware mappings. Each transformation should preserve semantics, which makes \emph{circuit equivalence} a central verification task: checking whether two circuits generated at different stages of the toolchain induce the same unitary, possibly up to a global phase.

The exact circuit equivalence problem is \textbf{NQP}-hard~\cite{tanaka2010exact}, which makes a polynomial-time exhaustive algorithm unlikely. Work in this area has focused on two main approaches. The first is \emph{equational reasoning}, which chooses a symbolic representation together with an equational theory and searches for a rewriting path between circuit representations~\cite{coecke2017picturing,Coecke2011interacting,amy2018towards,vilmart2021structure,clement2022lov}. This approach can be efficient when short proofs exist, but its completeness and confluence properties degrade outside restricted fragments. The second is \emph{model-based reasoning}, which encodes equivalence as a decision problem on a semantic model and solves it with a complete, though often expensive, procedure~\cite{burgholzer2020advanced,sistla2023symbolic,mei2024equivalence,chen2023pldi,chen2023cav,chen2025popl}.

In this work, we study both approaches in concrete form. For equational reasoning, we use the \emph{path-sum calculus}~\cite{amy2018towards}, an equational system that represents circuits as algebraic sums over computational paths. This enables algebraic rewriting rules to simplify path-sum formulae and prove equivalence when a short derivation exists. For model-based reasoning, we introduce a new use of \emph{weighted model counting} (WMC).
To highlight the contrast, we present two simple examples. The path-sum method shows how algebraic rewrites can prove that $\Hg(1) \Hg(1)$ equals the identity, where $\Hg(1)$ means applying the Hadamard gate H gate to the first qubit. We then illustrate our WMC encoding using the same example for simplicity, noting that the same encoding also applies in cases where algebraic rewriting alone is insufficient.

\begin{example}[Path-sum reasoning]
Consider the single-qubit circuit $C = \Hg(1) \Hg(1)$, where $\Hg$ is the Hadamard gate.
The path-sum semantics of $\Hg(1)$ is
\[
\Hg(1) \colon \kett{x} \mapsto \tfrac{1}{\sqrt{2}} \sum_{y \in \{0,1\}} (-1)^{x \cdot y} \kett{y}.
\]
This matches the standard semantics of $\Hg(1)$, which maps $\kett{0}$ to $\tfrac{1}{\sqrt{2}} (\kett{0} + \kett{1})$ and $\kett{1}$ to $\tfrac{1}{\sqrt{2}} (\kett{0} - \kett{1})$. The path-sum formula describes how the basis state $\kett{x}$ contributes to each basis state $\kett{0}$ and $\kett{1}$ after the gate execution.

Composing two $\Hg$ gates (see formula \ref{comp} for formal definition) yields

\begin{equation}
\Hg \Hg \colon \kett{x} \mapsto \tfrac{1}{2} \sum_{y,z \in \{0,1\}} (-1)^{x \cdot y + y \cdot z} \kett{z}.\label{ex:HH}
\end{equation}

Using the [\Hg\Hg] rewrite rules that will be introduced in Section~\ref{sec:background:pathsums}, the above path-sum simplifies to $\kett{x} \mapsto\kett{x}$, proving that $\Hg \Hg$ is the identity.
\end{example}

Before giving the corresponding WMC example, we recall the definition of weighted model counting.

\begin{definition}[Weighted Model Counting]
Let $\Gamma$ be a Boolean formula over variables $V$, and let 
$W : V \times \{0,1\} \to \mathbb{A}$ be a weight function that assigns
to each literal $(v,b)$ a weight $W(v,b)$ in some commutative semiring $\mathbb{A}$.
The \emph{weighted model count} of $(\Gamma,W)$ is defined as
\[
\wmc(\Gamma,W) \;=\;
\sum_{\tau \models \Gamma} \; \prod_{v \in V} W\!\big(v,\tau(v)\big),
\]
where the sum ranges over all satisfying assignments $\tau : V \to \{0,1\}$ of $\Gamma$.
\end{definition}

\begin{example}[Reducing Path-Sum Reasoning to Weighted Model Counting]
Rather than applying rewrite rules, we encode the problem of checking whether the path-sum for $\Hg(1) \Hg(1)$ equals the identity as a weighted model counting instance.

We start with the original variables $x,y,z$ and introduce two auxiliary Boolean variables $s_1$ and $s_2$, representing the products $x \cdot y$ and $y \cdot z$, respectively. The Boolean formula fed into the weighted model counter is
\[
\Gamma \;\equiv\; (s_1 \leftrightarrow (x \wedge y)) \;\wedge\; (s_2 \leftrightarrow (y \wedge z)) \;\wedge\; (x \leftrightarrow z).
\]

The weight function is defined as follows. For Boolean variable $b\in\{0,1\}$,
\[
W(v,b) =
\begin{cases}
1 & \text{if $v \in \{x,y,z\}$}, \\[4pt]
(-1)^b & \text{if $v \in \{s_1,s_2\}$}.
\end{cases}
\]

Below, we explain the idea behind the encoding $\Gamma$ and weight function $W$. First, observe that all integer variables in Equation~(\ref{ex:HH}) are over the domain $\{0,1\}$, in this introduction, we slightly abuse the notation and treat them simultaneously as Boolean and integer variables. So the two Boolean constraints $s_1 \leftrightarrow (x \wedge y)$ and $s_2 \leftrightarrow (y \wedge z)$ also have a corresponding integer interpretation $s_1=x\cdot y$ and $s_2 = y \cdot z$.
Hence the amplitude of $\kett{z}$ in Equation~(\ref{ex:HH}) can be rewritten as
\[
\tfrac{1}{2} \sum_{y,z \in \{0,1\}} (-1)^{x \cdot y + y \cdot z} 
\;=\;
\tfrac{1}{2} \sum_{y,z \in \{0,1\}} (-1)^{s_1} \cdot (-1)^{s_2}.
\]
From the right-hand side of the above equation, it suffices to consider only the contributions of $s_1$ and $s_2$ to the amplitude of the target basis $\kett{z}$. Accordingly, we assign the weight $(-1)^{b_i}$ to the corresponding literal $(s_i,b_i)$, while all other literals are assigned weight one.

Finally, the constraint $x\leftrightarrow z$ ensures that we only consider amplitudes where $\kett{x}$ contributes to itself in the path-sum, i.e., the diagonal entries of the corresponding unitary matrix. It follows that the path-sum equals the identity up to global phase if and only if the absolute value of the weighted count is $2^1 \times e^{i\theta}$, where $e^{i\theta}$ is the \emph{global phase}, and this example has only one qubit, which matches the diagonal sum of a $2 \times 2$ identity matrix.

\end{example}

To summarize, we propose a hybrid method that integrates equational reasoning with model-based reasoning. Our method
applies reductions whenever possible and invokes the WMC-based decision procedure on the residual path-sum, yielding a complete semantic check up to a global phase.
Our approach is built on the \emph{path-sum} framework~\cite{amy2018towards}, where the unitary evolution of a quantum system from state $A$ to state $B$ is expressed as a weighted sum over all possible paths. 
Our main contributions are:

\begin{itemize}
    \item A novel weighted model counting (WMC) encoding of the path-sum equivalence problem in the sense of Amy~\cite{amy2018towards}, enabling a complete semantic check.
    \item A hybrid verification tool that combines path-sum reductions with the WMC-based procedure. 
    \item An implementation and experimental evaluation on standard benchmarks, demonstrating that the hybrid approach consistently outperforms either reductions or WMC alone, and improves over simple strategy combinations.
\end{itemize}

\paragraph{Related work.}
Path-sum reduction produces phase polynomials with symbolic path conditions, a structured algebraic object that cannot be encoded using any existing WMC formulation, including Quokka\#’s per-gate encoding for quantum circuit equivalence testing~\cite{mei2024simulating,mei2024equivalence,mei2024disentangling}.

Our contribution is a new encoding that translates these reduced forms into weighted Boolean variables and constraints. This allows WMC to operate after reduction, which was previously impossible. Without this encoding, model counting cannot even express the intermediate objects produced by path-sum.

Two main lines of research exist for quantum circuit equivalence checking: \emph{equational reasoning} and \emph{model-based approaches}.  

Equational reasoning frameworks, such as path-sum~\cite{amy2018towards,amy2023complete,amy2025POPL,chareton2021automated,ricciardi2025quantum} and the ZX-calculus~\cite{coecke2017picturing,duncan2009graph,jeandel2018lics,duncan2010rewriting}, establish equivalence by algebraic or diagrammatic rewriting. These methods are sound, and in principle, they can be made complete. Yet completeness forces reductions to raw matrix operations, which are only of theoretical interest. Practical tools, therefore, trade completeness for efficiency, applying partial but efficient reasoning procedures.

Model-based approaches, by contrast, construct explicit symbolic models of quantum states and circuits using decision diagrams~\cite{burgholzer2020advanced,hong2022tensor,hong2024equivalence,roland2022}, tree automata~\cite{chen2023cav,chen2025cacm,chen2025popl,chen2025tacas}, SMT~\cite{Bauer2023symQV,chen2023cade}, or weighted model counting~\cite{mei2024simulating,mei2024equivalence,mei2024disentangling}. These techniques are typically complete. They may be slower than equational methods on simple circuits, but they scale more robustly on challenging cases, often outperforming equational methods when completeness becomes impractical.

Combining equational reasoning with model-based approaches offers the potential for better synergy.  
An early attempt was made by \qcec~\cite{burgholzer2020advanced}, which combines ZX-calculus with decision diagrams in a portfolio approach. 
This design, however, runs two solvers in parallel without interaction.  
Another approach to combine both paradigms appeared in the context of symbolic variational circuits~\cite{peham2023equivalence}. This approach complements the incomplete ZX approach with a complete method,
but this approach has not yet been evaluated on general circuit equivalence checking. 



\section{Background}
\label{sec:background}

\subsection{Quantum preliminaries}
\label{sec:background:quantum}

We present only the essentials due to space constraints; for a comprehensive treatment, see~\cite{nielsen2002quantum}.

A one-\emph{qubit} quantum state is a unit vector $\smat{\alpha_0\\\alpha_1}\in\mathbb{C}^2$ written in Dirac notation as
$\alpha_0\kett{0}+\alpha_1\kett{1}$ with $|\alpha_0|^2+|\alpha_1|^2=1$. An $n$-qubit quantum state can be written as $\kett{\psi}=\sum_{x\in\{0,1\}^n}\alpha_x\kett{x}$, with the sum of absolute square value of all probabilistic amplitudes $\sum_{x\in\{0,1\}^n}|\alpha_x|^2=1$.
The dynamics of a quantum state are governed by \emph{unitary operators}.
A unitary $U$ is a linear map on $\mathcal{H}$ with $U^\dagger U=I$, where
$U^\dagger$ denotes the \emph{conjugate transpose} of $U$. Applying the unitary gate $U$ to $\kett{\psi}$ produces the state $U\kett{\psi}$.
Elementary quantum \emph{gates} are fixed unitaries acting on a small number of qubits (one or two).  
Common examples include the Hadamard gate
$\Hg(1)\kett{x}=\tfrac{1}{\sqrt{2}}(\ket{0}+(-1)^x\kett{1})$, the controlled-NOT gate
$\CNOT(1,2)\kett{x,z}=\kett{x,x\oplus z}$, the rotation gate $\Rz(\theta)(1)\kett{x} = e^{i\theta\cdot x}\kett{x}$, and the controlled-rotation gate $\CRz(\theta)(1,2)\kett{x,z} = e^{i\theta\cdot x\cdot z}\kett{x,z}$.


A \emph{quantum circuit} is a finite sequence of gates; each gate is a unitary acting on one or more qubits, and the whole circuit denotes a unitary matrix $U\in\mathbb{C}^{2^n\times 2^n}$.
Parallel composition of two circuits corresponds to their \emph{tensor product} and sequential composition to their \emph{matrix multiplication}.
Two circuits $C_1,C_2$ on the same $n$ qubits are  \emph{equivalent} if $U[C_1]=U[C_2]$, where $U[C]$ means the unitary matrix of the circuit $C$. They are \emph{equivalent up to global phase} if $U[C_1]=e^{i\theta}U[C_2]$ for some real $\theta$.
This is the standard notion for circuit equivalence checking.
Equivalently, to check if two circuits $C_1$ and $C_2$ are equivalent, it suffices to verify that $U[C_2]U[C_1]^\dagger=I$. They are equivalent up to global phase iff $U[C_2]U[C_1]^\dagger=e^{i\theta}I$ for some real $\theta$.

\subsection{Path-sum formalism}
\label{sec:background:pathsums}
\emph{Path-sums} provide a symbolic representation of a quantum circuit's semantics as sums over Boolean ``paths''~\cite{amy2018towards,amy2023complete}.
Formally, for a sequence of state variables $\vec{x}=\{x_1,\ldots,x_n\}$ (where $n$ is the number of qubits), a path-sum is a triple $\mathcal{P} =\langle \vec{y},\,\Phi,\,\mathcal{O}\rangle$ where
(i) $\vec{y}=\{y_1,\ldots, y_m\}$ is the set of Boolean \emph{path variables},
(ii) $\Phi:\mathbb{B}^{|\vec{x}|+|\vec{y}|}\mapsto \mathbb{R}$ is a phase polynomial with real coefficients, and
(iii) $\mathcal{O}:\mathbb{B}^{|\vec{x}|+|\vec{y}|}\mapsto \mathbb{B}^{|\vec{x}|}$ is an output function that maps path and state variables to a Boolean vector describing the output state. 
Its action on a basis state $\kett{\vec{x}}$ is given as
\begin{equation}
\label{psgen}
\ket{\vec{x}}
~\mapsto~
\frac{1}{\sqrt{2^{|\vec{y}|}}}\!\sum_{\vec{y}\in\{0,1\}^{|\vec{y}|}}
e^{i\,\Phi(\vec{x},\vec{y})}\;\ket{\mathcal{O}(\vec{x},\vec{y})}.
\end{equation}

Elementary gates admit concise path-sum rules. For example,
writing $z_k$ to denote the $k$th element of a vector $\vec z$, we have:
\[
\begin{aligned}
\Hg(t):\quad & \langle \vec{y},\Phi,\mathcal{O}\rangle
~\mapsto~
\big\langle \vec{y}\!\cup\!\{y_f\},\;\Phi+ \widetilde{x_t} \widetilde{y}_f
\pi,\;\mathcal{O}[x_t\!\leftarrow\! y_f]\big\rangle,& (y_f \text{ is fresh})\\
\CNOT(c,t):\quad &
\langle \vec{y},\Phi,\mathcal{O}\rangle
~\mapsto~
\langle \vec{y},\Phi,\mathcal{O}[x_t\!\leftarrow\! x_c\oplus x_t]\rangle,\\
\Rz(\theta)(t):\quad & 
\langle \vec{y},\Phi,\mathcal{O}\rangle
~\mapsto~
\langle \vec{y},\Phi+\widetilde{x_t} \theta,\mathcal{O}\rangle,\\
\CRz(\theta) (c,t):\quad &
\langle \vec{y},\Phi,\mathcal{O}\rangle
~\mapsto~
\langle \vec{y},\Phi+\widetilde{x_c}\widetilde{x_t} \theta,\mathcal{O}\rangle.
\end{aligned}
\]
where $t$ denotes the target qubit, $c$ denote the control qubit, $\widetilde{b}$ denotes \emph{bool-to-int casting}, interpreting XOR $\oplus$ as \(\widetilde{b_i \oplus b_j} := \widetilde{b_i} + \widetilde{b_j} - 2 \cdot \widetilde{b_i} \cdot \widetilde{b_j}\) and NOT $\neg$ as $\widetilde{\neg b} := 1-\widetilde{b}$ \footnote{In practice, we omit the explicit casting notation  for  atomic binary integer values}
and 
\[
\mathcal{O}[x_t\!\leftarrow\!\sigma](\vec{x},\vec{y})_j =
\begin{cases}
\mathcal{O}(\vec{x},\vec{y})_j, & \text{if } t\neq j,\\
\sigma, & \text{if } t=j.
\end{cases}
\]

\begin{example}
    The path-sum triple for the circuit $\Hg(1)\Hg(1)$ is constructed as follows:
    \[
    \langle \emptyset,0,\kett{x}\rangle
    \xrightarrow[]{\Hg(1)} \langle \{y_0\},xy_0\pi,\kett{y_0}\rangle
    \xrightarrow[]{\had(1)}
    \langle \{y_0,y_1\},xy_0\pi+y_0y_1\pi,\kett{y_1}\rangle.
    \]
    Returning to Dirac notation, the action of the circuit on each basis state $\kett{x}$ is
    \[
    \kett{x} \mapsto \frac{1}{\sqrt{2^{|\{y_0,y_1\}|}}}\sum_{y_0,y_1\in\{0,1\}} e^{i(xy_0\pi+y_0y_1\pi)}\kett{y_1}
    =
    \frac{1}{2}\sum_{y_0,y_1\in\{0,1\}} (-1)^{(xy_0+y_0y_1)}\kett{y_1},
    \]
    which coincides with the expression in Example~\ref{ex:HH}.\qed
\end{example}




Composition of circuits corresponds to the symbolic composition of path-sums by substituting outputs and adding phases as follows.
\begin{equation}
\label{comp}
\langle \vec{y}',\Phi',\mathcal{O}'\rangle \circ \langle \vec{y},\Phi,\mathcal{O}\rangle
\;=\;
\big\langle \vec{y}\!\cup\!\vec{y}',\;\Phi(\vec{x},\vec{y})+\Phi'(\mathcal{O}(\vec{x},\vec{y}),\vec{y}'),\;
\mathcal{O}'(\mathcal{O}(\vec{x},\vec{y}),\vec{y}')\big\rangle
\end{equation}

The path-sum of the identity circuit is $\langle \emptyset, 0, \kett{x}\rangle$.
Any quantum circuit can be expressed as a path-sum by starting from this identity form and composing it step by step with each gate in the circuit.

However, this gate-by-gate construction can yield different formulas representing the same vector semantics as in Equation~\eqref{psgen}. For instance, applying two consecutive Hadamard gates to the same qubit yields  the path-sum 
\[\pasum{\vec{y}}{\Phi}{\mathcal{O}}{}  \xrightarrow[]{\Hg(i)\Hg(i)}\pasum{\vec{y}\cup \{y_0,y_1\}}{\Phi + x_iy_0\pi + y_0y_1\pi}{\mathcal{O}[x_i\leftarrow y_0][x_i\leftarrow y_1]}{} \]
Here, $y_0$ and $y_1$ are fresh variables not appearing in the original path-sum.

Since the Hadamard operation is involutive, $\Hg\Hg$ is equivalent to the identity, we also have the equivalent and much simpler path-sum
\[\pasum{\vec{y}}{\Phi}{\mathcal{O}}{}  \xrightarrow[]{\Hg(i) \Hg(i)}\pasum{\vec{y}}{\Phi}{\mathcal{O}}{} \]

To formalize this equivalence, the path-sum framework is equipped with an \emph{equational theory}, that is, a collection of reduction rules generating an equivalence relation compatible with the path-sum semantics.

\begin{figure}[htbp]
    \centering
\begin{mdframed}[linecolor=black, linewidth=.5pt]
\begin{prooftree}
     \AxiomC{$y_0 \notin (\var{\mathcal{O}}\cup \var{Q} \cup \var{R})$}
    \AxiomC{$y_1 \notin \var{Q}$}
    \LeftLabel{[\Hg\Hg]}
    \BinaryInfC{$\pasum{\vec{y}\cup\{y_0,y_1\}}{(\frac{1}{2}y_0(y_1 + Q)+R)\cdot 2\pi}{\mathcal{O}}{} \equiv \pasum{\vec{y}}{R[y_1 \leftarrow \overline{Q} ]\cdot 2\pi}{\mathcal{O}[y_1 \leftarrow  \overline{Q}]}{}$} 

\DisplayProof\vskip 1em
    \AxiomC{$y_0 \notin (\var{\mathcal{O}}\cup \var{Q} \cup \var{R})$}
    \LeftLabel{[$\omega$]}
    \UnaryInfC{$\pasum{\vec{y}\cup\{y_0\}}{(\frac{1}{4}y_0 + \frac{1}{2}y_0Q+R)\cdot 2\pi}{\mathcal{O}}{} \equiv \pasum{\vect{y}}{(\frac{1}{8}-\frac{1}{4}\overline{Q}+R)\cdot 2\pi}{\mathcal{O}}{}$}
\end{prooftree}
\end{mdframed}
\caption{The path-sum reduction rules~\cite{amy2018towards}, where $Q$ is an integer formula and $R$ a real formula. $\overline{Q}$ computes the modulo 2 integer value, and we also abuse the notation to denote the corresponding Boolean value, i.e., one maps to true and zero to false. We refer the readers to~\cite{amy2018towards} for the $\overline{Q}$ construction.}\label{fig:path-sum-reduction-rules}
\end{figure}
The main reduction rules are shown in Fig.~\ref{fig:path-sum-reduction-rules}.
Rule [\Hg\Hg]
 symbolically identifies and cancels pairs of paths that contribute equal amplitudes but with opposite phases.
Similarly, rule [$\omega$] captures a different kind of symbolic simplification: when two amplitude vectors differ by a $\pi/2$ phase, their sum corresponds to a \emph{bisector} direction with normalized amplitude.\footnote{If two complex amplitudes $a$ and $a e^{i\pi/2}$ have equal magnitude but differ by a phase of $90^\circ$, their sum is $a(1+i) = \sqrt{2}a e^{i\pi/4}$. The resulting vector points in the direction halfway between them in the complex plane, called the \emph{bisector}.} The rule abstracts this geometric addition at the symbolic level. We refer readers to \cite{amy2018towards} for the detailed derivations of these equations. Readers may safely skip them on a first reading without losing track of our main technical contribution.

For \emph{Clifford circuits},\footnote{Clifford circuits are those generated by the Hadamard, Phase, and CNOT gates. They form the Clifford group, which is efficiently classically simulable~\cite{gottesman1998heisenberg}.} the reduction rules from Fig.~\ref{fig:path-sum-reduction-rules} are both \emph{complete} (any two equivalent circuits can be reduced to the same normal form)  
and \emph{confluent} (the reduction process always yields a unique normal form, regardless of the order of rule applications). 
Beyond that, for example, in the $Z^*$ fragment,\footnote{Circuits in gate set $\{H, CNOT, R_k=\left(\begin{smallmatrix} 1 & 0 \\ 0 & e^{\nicefrac{2\pi i}{2^k}} \end{smallmatrix}\right)\}$.} completeness can be obtained but confluence is lost~\cite{vilmart2023rewriting}. This means that while equivalence is always provable in principle, there is no deterministic strategy to find a proof, and practical verification must rely on heuristics.

In this paper, we present a complete method for handling quantum circuits generated by any set of gates, including the $Z^*$ fragment.  
We first apply the reduction rules (e.g., those in Fig.~\ref{fig:path-sum-reduction-rules}) to simplify the formula as far as possible.  
For the residual path-sum, we then show in the next section how equivalence checking (up to global phase) can be reduced to a weighted model counting problem.


\section{Technical Contribution}
\label{sec:technical}

Our goal is to verify the equivalence of two $n$-qubit circuits $C_1,C_2$. 
Due to reversibility of the circuits,
this reduces to checking whether 
\[
C := C_1C_2^\dagger
\]
is equal to $I$ (exact equivalence) or to $e^{i\theta} I$ for some real $\theta$ (equivalence up to global phase).
We construct the path-sum $\ps(C)$ and then proceed in two steps:
\begin{enumerate}
\item apply path-sum reductions (Fig.~\ref{fig:path-sum-reduction-rules}) to simplify $\ps(C)$, and 
\item encode the residual instance as a weighted model counting (WMC) problem and decide it exactly.
\end{enumerate}

Step~1 largely follows the standard path-sum calculus, but requires some refinements. 
In particular, reduction rules depend on recognizing syntactic patterns in path-sum formulae, yet many equivalent expressions do not match these patterns directly. 
In practical use, one needs to add additional rules to reorganize the phase polynomials and Boolean output functions to enable pattern matching. 
We also slightly relaxed the constraints for $\Rz(\theta)$ and $\CRz(\theta)$ gate constructions to permit arbitrary rotation angles. 
These relaxations integrate smoothly with the path-sum formalism of~\cite{amy2018towards}, but, as expected, the reduction rules cannot handle interference on arbitrary angles and therefore remain incomplete.
This motivates Step~2, where we derive a sound and complete equivalence-checking procedure for path-sum formulae through a reduction to WMC.


\begin{example}
\label{runex1}
We illustrate a simple case of gate simulation. 
The controlled rotation gate $\CRz(\pi)$ (Fig.~\ref{runex:intro}) can be simulated by a $\CNOT$ gate conjugated with Hadamard gates on its target qubit (Fig.~\ref{runex:simu}).
We refer to the circuit in Fig.~\ref{runex:intro} as $C_1$ and the one in Fig.~\ref{runex:simu} as $C_2$.
The two circuits $C_1$ and $C_2$ are \emph{semantically equivalent}.
\begin{figure}
\centering
\begin{subfigure}{.5\textwidth}
\centering
\begin{quantikz}[row sep=0.3cm,column sep=.3cm,wire types={q,q}]
\lstick{$\ket{x}{}$}&
 &\ctrl{1}&&
\\ 
\lstick{$\ket{z}{}$}&
 &\gate{\zz}&&
\\ 
\end{quantikz}
\label{runex:zrot}

\caption{The controlled-rotation $C_1:=\CRz(\pi)(1,2)$}
\label{runex:intro}
\end{subfigure}
 \hspace{1cm}
 \begin{subfigure}{.4\textwidth}
\centering
\begin{quantikz}[row sep=0.3cm,column sep=.3cm,wire types={q,q}]
\lstick{$\ket{x}{}$}
&
 &\ctrl{1}&&
\\ 
\lstick{$\ket{z}{}$}&\gate{\had}
 &\targ{}
  &\gate{\had}
 &
\end{quantikz}
\caption{The controlled Z rotation:  Clifford simulation  through circuit $C_2 :=\had(2) \CNOT(1,2) \had(2)$}
\label{runex:simu}
 \end{subfigure}
\caption{Running example, simulating Control Z rotation with Clifford gates}
\end{figure}
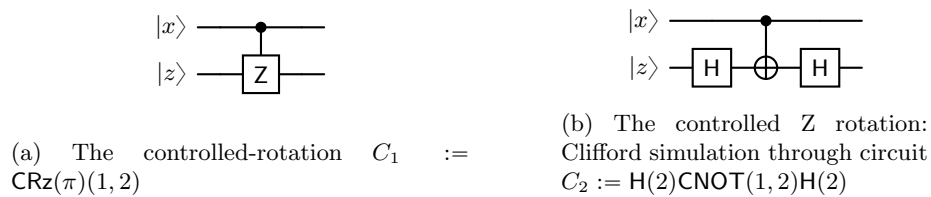
To check the equivalence between $C_1$ and $C_2$,
since $U[C_2]^\dagger = U[C_2]$ based on its definition, Lemma~\ref{lem:identity}
instantiates as:
\begin{equation}C_1
\equiv C_2 \textit{ iff for all }\vec x \in \{0,1\}^2, \ps(C_1C_2^{\dagger})\ket{\vec x}{} = \ket{\vec x}{}.
\label{eq:formal_problem_instance_runex}\end{equation}
Based on the composition rule from Equation~(\ref{comp}), we have the composed circuit in path-sum formalism as follows:
\begin{align}
\ps(C_1;C_2)
&= \pasum{\emptyset}{0}{\kett{x,z}}{}\notag\\&\xrightarrow[]{\mathrm{CR}_Z(\pi)(1,2)}
   \pasum{\emptyset}{\pi x z}{\kett{x,z}}{}\notag\\[2pt]
&\xrightarrow[]{H(2)}
   \pasum{\{y_0\}}{\pi\,(x z + z y_0)}{\kett{x,y_0}}{}\notag\\[2pt]
&\xrightarrow[]{\mathrm{CNOT}(1,2)}
   \pasum{\{y_0\}}{\pi\,(x z + z y_0)}{\kett{x,\,x\oplus y_0}}{}\notag\\[2pt]
&\xrightarrow[]{H(2)}
   \pasum{\{y_0,y_1\}}{\pi\!\big(x z + z y_0 + \widetilde{(x\oplus y_0)}\,y_1\big)}{\kett{x,y_1}}{}\notag\\[2pt]
&= \pasum{\{y_0,y_1\}}{\pi\!\big(x z + z y_0 + (x + y_0 - 2 x y_0)\,y_1\big)}{\kett{x,y_1}}{}\label{eq:ps_before}\\[2pt]
&\equiv
   \pasum{\{y_0,y_1\}}{\pi\!\big(x z + z y_0 + (x + y_0)\,y_1\big)}{\kett{x,y_1}}{} \quad (\mathrm{mod}\;2\pi).\label{eq:ps_after}
\end{align}
 Note that for the last step,
 since $e^{-2xy_0\pi} = 1$, we can further reduce the resulting path-sum from Equation~(\ref{eq:ps_before}) to Equation~(\ref{eq:ps_after}).
 To apply [\Hg\Hg] reduction, one has to rearrange the term $(x z + z y_0 + (x + y_0)\,y_1)\cdot \pi$ to be in the form of the polynomial
 $(\frac{1}{2}y_0(y_1 + z) + \frac{1}{2}x(y_1 + z))\cdot 2\pi$,
 where $Q=z$, $R=\frac{1}{2}x(y_1 + z)$, $\mathcal{O}(\vec x, \vec y)=(x,y_1)$.
 So $y_0\notin var(Q)\cup var(R) \cup var(\mathcal{O})$,
 the phase polynomial is then reduced to be $\pasum{\emptyset}{R[y_1\leftarrow\bar Q]\cdot 2\pi}{\mathcal{O}[y_1\leftarrow \bar Q]}{}$,
 where $R[y_1\leftarrow\bar Q]\cdot 2\pi=\frac{1}{2} x(z+z)\cdot 2\pi \equiv 0 \quad(\mathrm{mod} \;2\pi)$ and $\mathcal{O}(\vec x, \vec y)[y_1\leftarrow z]=\kett{x,z}$.

 Thus we have the reduction:
Equation~$(\ref{eq:ps_after}) \xrightarrow[]{[\Hg\Hg]}\pasum{\emptyset}{0}{\kett{x,z}}{}$.\qed

\end{example}

In \cite{amy2018towards}, completeness is established only for the Clifford fragment of the calculus. To extend it to general quantum circuits, we reduce the problem to weighted model counting (WMC). 


\subsection{Encoding path-sum equivalence into WMC}
\label{sec:technical_contrib:wmc}
We reduce equivalence to a check on diagonal amplitudes. For
\({C}=C_2C_1^\dagger\) and basis \(\kett{\vec x}\), we assume $\ps({C})$ is
\[
\kett{\vec x}
~\mapsto~\frac{1}{\sqrt{2^{|\vec{y}|}}}\sum_{\vec y\in\{0,1\}^{|\vec{y}|}}
e^{i\,\Phi(\vec x,\vec y)}\ket{\mathcal{O}(\vec x,\vec y)}.
\]
and slightly abuse notation and write $\ps(C)$ also for its \emph{matrix representation}. Formally, it defines the matrix $\frac{1}{\sqrt{2^{|\vec{y}|}}}\sum_{\vec y\in\{0,1\}^{|\vec{y}|}}
e^{i\,\Phi(\vec x,\vec y)}\kett{\mathcal{O}(\vec x,\vec y)}\bra{\vec x}$.

Because $\ps(C)$ is unitary, for each basis state $\kett{\vec x}$, the output $\ps(C)\kett{\vec x}$ has $L_2$-norm equals 1. 
Therefore $\ps(C)\kett{\vec x}=\kett{\vec x}$ iff its \emph{diagonal amplitude} satisfies 
$\langle \vec x|\ps(C)|\vec x\rangle=1$. 
For equivalence up to a global phase, it suffices to require 
$\langle \vec x|\ps(C)|\vec x\rangle=e^{i\theta}$ 

\begin{lemma}\label{lem:identity}
$\ps(C)$ is the identity if and only if, 
for every basis state $\kett{\vec x}$, its diagonal amplitude satisfies
\[
\langle\vec x|\ps(C)|\vec x\rangle = 1.
\]
\end{lemma}
\begin{corollary}
\label{corollary:eq}
Let $\ps(C)$ act on $n$ qubits. 
Then $\ps(C)$ is equivalent to identity up to global phase if and only if
\[
\sum_{\vec{x}\in\{0,1\}^n} \langle \vec{x}\,|\,\ps(C)\,|\,\vec{x}\rangle \;=\; 2^n\cdot e^{i\theta} \text{ for a given }\theta.
\]
\end{corollary}
\begin{proof}
The ``only if'' direction is immediate: if $\ps(C)=e^{i\theta} \cdot I$, for every $\kett{\vec x}$,
$\langle \vec x|\ps(C)|\vec x\rangle = e^{i\theta}\cdot\langle \vec x|\vec x\rangle = e^{i\theta}$, thus $\sum_{\vec{x}\in\{0,1\}^n} \langle \vec{x}\,|\,\ps(C)\,|\,\vec{x}\rangle \;=\; 2^n\cdot e^{i\theta}$.

For the ``if'' direction, suppose $\sum_{\vec x\in \{0,1\}^n}\langle \vec x|\ps(C)|\vec x\rangle=2^n\cdot e^{i\theta}$.  
Since $\ps(C)$ is unitary, for each basis state $\vec x$ we have 
$|\langle \vec x|\ps(C)|\vec x\rangle|\le 1$. Therefore,
\[
2^n
= \Big|\sum_{\vec x\in \{0,1\}^n}\langle \vec x|\ps(C)|\vec x\rangle\Big|
\;\le\; \sum_{\vec x\in\{0,1\}^n}\big|\langle \vec x|\ps(C)|\vec x\rangle\big|
\;\le\; \sum_{\vec x\in\{0,1\}^n} 1
= 2^n.
\]
Since the first and last terms of the chain are equal, both inequalities must in fact be equalities, Thus $|\langle \vec x|\ps(C)|\vec x\rangle|=1$ for every $\vec x$ and since the modulus of total sum is the maximum $2^n$, 
that global phase must be the same, i.e.\ $\langle \vec x|\ps(C)|\vec x\rangle=c$ for all $\vec x$. 
So $\langle \vec x|e^{-i\theta}\cdot\ps(C)|\vec x\rangle=1$.
Based on Lemma~\ref{lem:identity}, 
we conclude $e^{-i\theta}\cdot\ps(C)= I$ thus $\ps(C)= e^{i\theta}\cdot I$.\qed
\end{proof}

Next we show how to use WMC to compute the sum of the diagonal values $\sum_{\vec{x}\in\{0,1\}^n} \langle \vec{x}\,|\,\ps(C)\,|\,\vec{x}\rangle$.
For the encoding, we first rewrite the \emph{phase polynomial} $\Phi(\vec{x},\vec{y})$ of $\ps(C)$ to a finite sum of monomial terms
\[
\Phi(\vec{x},\vec{y})=\sum_{j\in J} \Mon_j,
\qquad
\Mon_j = k_j\cdot \boolcond_j(\vec{x},\vec{y}).
\]
Here $J$ is the set of monomial index, each $\Mon_j$ is the product of a real coefficient $k_j$ 
and a Boolean condition $\boolcond_j(\vec{x},\vec{y})$. 
The Boolean condition indicates \emph{when} this monomial contributes: 
it evaluates to $1$ if the assignment $(\vec{x},\vec{y})$ activates the term, 
and $0$ otherwise.

\begin{example}\label{runex3}
Considering the path-sum in Equation~(\ref{eq:ps_before}) of Example~\ref{runex1}:
    \[\ps(C_2;C_1^\dagger) = \pasum{\{y_0, y_1\}}{{(xz+zy_0+{(x + y_0 - 2xy_0)}y_1)\cdot \pi}}{\ket{x,y_1}{}}{}\]
Here, the phase term can be written as ${(y_0y_1+y_0z+xy_1+xz-2xy_0y_1)\cdot \pi}$, which expands as a  monomials sum of 5 monomial terms, with the following definition and constraints:
\[\begin{array}{rcl|rcccl|rcccl}
\multicolumn{3}{c|}{\textbf{Monomial definition}}& \multicolumn{5}{c|}{\textbf{Boolean constraint}}& \multicolumn{5}{c}{\textbf{Real Coefficient}}\\
\Mon_1    &:= &  {y_0y_1\pi}&\boolcond_1(\vec{x},\vec{y})&:=& y_0\wedge y_1
&&&k_1&:=& \pi\\
  \Mon_2  &:= &  y_0z\pi&\boolcond_2(\vec{x},\vec{y})&:=& y_0\wedge z
  &&&k_2&:=& \pi\\
\Mon_3    &:= & xy_1\pi&\boolcond_3(\vec{x},\vec{y})&:=& x\wedge y_1
&&&k_3&:=& \pi\\
  \Mon_4 &:= & {xz\pi}&\boolcond_4(\vec{x},\vec{y})&:=& x\wedge z
  &&&k_4&:=& \pi \\
  \Mon_5 &:= & {-2xy_0y_1\pi}&\boolcond_5(\vec{x},\vec{y})&:=& x\wedge y_0 \wedge y_1
  &&&k_5&:=& -2\pi
\end{array}\]
\qed
\end{example}

To bring this into WMC form, we proceed monomial by monomial. 
For each $j\in J$ we introduce a fresh Boolean variable $s_j$, 
referred to as a \emph{monomial variable}, together with a constraint
\[
\gamma_j :\quad s_j \leftrightarrow \boolcond_j(\vec{x},\vec{y}).
\]
Thus, $s_j$ explicitly represents whether the Boolean condition 
$\boolcond_j(\vec{x},\vec{y})$ holds. 
We then attach to $s_j$ a literal weight function
\[
W(s_j,0)=1,
\qquad
W(s_j,1)=e^{i k_j}.
\]
Intuitively, if the condition is false, the term contributes nothing, while if it is true, we multiply by the corresponding root of unity. 
We set the weight of all other literals to 1, as they are irrelevant to the sum of diagonal amplitude values.

From
\[
\ps(C)\ket{\vec x}
=\tfrac{1}{\sqrt{2^{m}}}\sum_{\vec y} e^{2\pi i\,\Phi(\vec x,\vec y)}\ket{\mathcal O(\vec x,\vec y)},
\]
we obtain
\[
\langle \vec x|\ps(C)|\vec x\rangle
=\tfrac{1}{\sqrt{2^{m}}}\sum_{\vec y} e^{2\pi i\,\Phi(\vec x,\vec y)}
\langle \vec x|\mathcal O(\vec x,\vec y)\rangle
=\tfrac{1}{\sqrt{2^{m}}}\sum_{\vec y:\,\mathcal O(\vec x,\vec y)=\vec x} e^{2\pi i\,\Phi(\vec x,\vec y)}.
\]
Here, as observed in~\cite[Lemma 4.1]{amy2018towards}, 
we have that by orthonormality 
$\langle \vec x|\mathcal O(\vec x,\vec y)\rangle$
equals the Kronecker delta $\delta_{\mathcal O(\vec x,\vec y),\,\vec x}$, equals $1$ if $x=O(\vec x,\vec y)$ and $0$ otherwise.
Hence, in the WMC encoding, we include the formula
\[
\beta: \;\; \bigwedge_{i=1}^{n}\big(\mathcal O(\vec x,\vec y)_i\leftrightarrow x_i\big),
\]
so that only diagonal paths are counted.
Finally, we form the global Boolean constraint
\[
\Gamma := \beta \wedge\ \bigwedge_{j\in J}\gamma_j,
\]
which enforces both the path restriction and the link between each $s_j$ 
and its underlying condition.

With this construction, the sum of diagonal  amplitude is obtained as

\[
\sum_{\vec{x}\in\{0,1\}^n} \langle \vec{x}\,|\,\ps(C)\,|\,\vec{x}\rangle \;=\;  \tfrac{1}{\sqrt{2^{m}}}\,\mathrm{WMC}(\Gamma,W),
\]
where $\mathrm{WMC}(\Gamma,W)$ denotes the weighted model count of $\Gamma$ 
under the literal weights $W$.
\begin{example}
The equivalence instance in Example~\ref{runex3} can be encoded as a weighted
model counting (WMC) problem. Using notations from Example~\ref{runex3}:
\[
\Gamma \;=\; \bigwedge_{j=1}^{5}\!\bigl(s_j \leftrightarrow \boolcond_j\bigr)\;\wedge\;\bigl(y_1 \leftrightarrow z\bigr).
\]
Since \(k_j=\pi\) for all \(j\in\{1,2,3,4\}\) and $k_5=-2\pi$, the weight function \(W\) is

\[
W(s_j,0)=1,\qquad 
W(s_j,1)=(-1)^{\delta (j,5)} 
\]

From the exhaustive enumeration in Table~\ref{tab:truth_table_runex}
(where shaded rows are the satisfying assignments), we obtain
\[
\mathrm{WMC}(\Gamma,W)
\;=\;
\sum_{\tau \models\Gamma} \;\prod_{j=1}^{5} W\!\bigl(s_j,\tau(s_j)\bigr)
\;=\; 8.
\]
Since \(m=2\) and \(n=2\),
\[
\sum_{\vec{x}\in\{0,1\}^2}
\bigl\langle \vec{x}\,\big|\,\ps(C)\,\big|\,\vec{x}\bigr\rangle
\;=\; 4 \;=\; \tfrac{1}{\sqrt{2^{2}}}\,\mathrm{WMC}(\Gamma,W).
\]
Therefore, the circuit $C_1$ is equivalent to the circuit $C_2$. Note that WMC solvers normally do not perform such exhaustive enumeration, we include it here to help the reader verify the correctness of the results.
\qed
\end{example}


\begin{table}[htbp]
    \centering
\[
\begin{array}{cccc||ccccc||c||c||cccc||ccccc||c||}
x&z&y_0&y_1&s_1&s_2&s_3&s_4&s_5&\prod W\!\big(v,\tau(v)\big)&\qquad&
x&z&y_0&y_1&s_1&s_2&s_3&s_4&s_5&\prod W\!\big(v,\tau(v)\big)\\
\hline\hline
\rowcolor{gray!15} 0&0&0&0&0&0&0&0&0&1&&1&0&0&0&0&0&0&0&0&1\\
0&0&0&1&0&0&0&0&0&1&&1&0&0&1&0&0&1&0&0&-1\\
\rowcolor{gray!15} 0&0&1&0&0&0&0&0&0&1&&1&0&1&0&0&0&0&0&0&1\\
0&0&1&1&1&0&0&0&0&-1&&1&0&1&1&1&0&1&0&1&1\\\hline\hline
0&1&0&0&0&0&0&0&0&1&&1&1&0&0&0&0&0&1&0&-1\\
\rowcolor{gray!15} 0&1&0&1&0&0&0&0&0&1&&1&1&0&1&0&0&1&1&0&1\\
0&1&1&0&0&1&0&0&0&-1&&1&1&1&0&0&1&0&1&0&1\\
\rowcolor{gray!15} 0&1&1&1&1&1&0&0&0&1&&1&1&1&1&1&1&1&1&1&1\\
\end{array}
\]
    \caption{Weight computation for the CZ case study}
    \label{tab:truth_table_runex}
\end{table}

\section{Implementation}
\label{sec:xps}
We implemented the proposed approach in a Python-based tool named \toolname (Quantum Path-sum Reduction and WMC Solver). For the implementation of the rule-based reduction, we use \symp and \symf to store and manipulate the phase polynomial part and Boolean function part of the pathsum formalism. For solving weighted model counting problems, we integrated GPMC~\cite{suzuki2017improvement}, which was selected due to its support for negative weights and good performance in the Model Counting Competition.\footnote{To meet the specific needs of our framework, we use an extended version of GPMC that handles complex-valued weights, available at \url{https://github.com/System-Verification-Lab/GPMC}.}

\begin{table}
\centering
\scalebox{0.60}{%

\begin{tabular}{cccc|cc|cc|cc}
\hline
\multicolumn{4}{c|}{benchmarks}                                                              & \multicolumn{2}{c|}{Hybrid}                                           & \multicolumn{2}{c|}{RR}                                               & \multicolumn{2}{c}{WMC}                                               \\ \hline
\multicolumn{1}{c|}{Benchmark Name}                                 & q   & G1   & G2   & time                                & result                                  & time                                & result                                  & time                                & result                                  \\ \hline
\multicolumn{1}{c|}{}                                                 & 32  & 34   & 33   & \cellcolor[HTML]{38FFF8}0.027         & \cellcolor[HTML]{38FFF8}$\bigtriangleup$ & \cellcolor[HTML]{38FFF8}0.027         & \cellcolor[HTML]{38FFF8}$\bigtriangleup$ & 0.08                                & $\bigtriangleup$                         \\
\multicolumn{1}{c|}{}                                                 & 64  & 66   & 65   & \cellcolor[HTML]{38FFF8}0.084         & \cellcolor[HTML]{38FFF8}$\bigtriangleup$ & \cellcolor[HTML]{38FFF8}0.084         & \cellcolor[HTML]{38FFF8}$\bigtriangleup$ & 0.212                               & $\bigtriangleup$                         \\
\multicolumn{1}{c|}{\multirow{-3}{*}{ghz}}                         & 128 & 130  & 129  & \cellcolor[HTML]{38FFF8}0.306         & \cellcolor[HTML]{38FFF8}$\bigtriangleup$ & \cellcolor[HTML]{38FFF8}0.306         & \cellcolor[HTML]{38FFF8}$\bigtriangleup$ & 0.617                               & $\bigtriangleup$                         \\ \hline
\multicolumn{1}{c|}{}                                                 & 16  & 160  & 108  & \cellcolor[HTML]{38FFF8}0.162         & \cellcolor[HTML]{38FFF8}$\bigcirc$       & \cellcolor[HTML]{38FFF8}0.162         & \cellcolor[HTML]{38FFF8}$\bigcirc$       & 0.286                               & $\bigcirc$                               \\
\multicolumn{1}{c|}{}                                                 & 32  & 320  & 219  & \cellcolor[HTML]{38FFF8}0.33          & \cellcolor[HTML]{38FFF8}$\bigcirc$       & \cellcolor[HTML]{38FFF8}0.33          & \cellcolor[HTML]{38FFF8}$\bigcirc$       & 0.778                               & $\bigtriangleup$                         \\
\multicolumn{1}{c|}{\multirow{-3}{*}{graphstate}}                   & 64  & 640  & 435  & \cellcolor[HTML]{38FFF8}1.068         & \cellcolor[HTML]{38FFF8}$\bigtriangleup$ & \cellcolor[HTML]{38FFF8}1.068         & \cellcolor[HTML]{38FFF8}$\bigtriangleup$ & 2.427                               & $\bigtriangleup$                         \\ \hline
\multicolumn{1}{c|}{}                                                 & 4   & 190  & 160  & \cellcolor[HTML]{38FFF8}0.213         & \cellcolor[HTML]{38FFF8}$\bigtriangleup$ & 0.162                               & $\times$                                & 0.446                               & $\bigtriangleup$                         \\
\multicolumn{1}{c|}{}                                                 & 5   & 499  & 409  & \cellcolor[HTML]{38FFF8}0.194         & \cellcolor[HTML]{38FFF8}$\bigtriangleup$ & \cellcolor[HTML]{38FFF8}0.194         & \cellcolor[HTML]{38FFF8}$\bigtriangleup$ & 0.743                               & $\bigtriangleup$                         \\
\multicolumn{1}{c|}{\multirow{-3}{*}{grover}}             & 6   & 1568 & 1400 & \textgreater{}200                     & TO                                      & \textgreater{}200                     & TO                                      & \cellcolor[HTML]{38FFF8}28.579        & \cellcolor[HTML]{38FFF8}$\bigtriangleup$ \\ \hline
\multicolumn{1}{c|}{}                                                 & 7   & 133  & 170  & 1.023                               & $\bigtriangleup$                         & 0.915                               & $\times$                                & \cellcolor[HTML]{38FFF8}1.031         & \cellcolor[HTML]{38FFF8}$\bigtriangleup$ \\
\multicolumn{1}{c|}{}                                                 & 9   & 171  & 209  & 0.974                               & $\bigtriangleup$                         & 0.793                               & $\times$                                & \cellcolor[HTML]{38FFF8}0.586         & \cellcolor[HTML]{38FFF8}$\bigtriangleup$ \\
\multicolumn{1}{c|}{\multirow{-3}{*}{qaoa}}                         & 11  & 209  & 270  & 1.636                               & $\bigtriangleup$                         & 1.442                               & $\times$                                & \cellcolor[HTML]{38FFF8}1.274         & \cellcolor[HTML]{38FFF8}$\bigtriangleup$ \\ \hline
\multicolumn{1}{c|}{}                                                 & 16  & 672  & 543  & \cellcolor[HTML]{38FFF8}1.818         & \cellcolor[HTML]{38FFF8}$\bigtriangleup$ & 0.738                               & $\times$                                & 1.932                               & $\bigtriangleup$                         \\
\multicolumn{1}{c|}{}                                                 & 20  & 1040 & 825  & \cellcolor[HTML]{38FFF8}20.122        & \cellcolor[HTML]{38FFF8}$\bigtriangleup$ & 1.445                               & $\times$                                & 20.706                              & $\bigtriangleup$                         \\
\multicolumn{1}{c|}{}                                                 & 22  & 1254 & 970  & \cellcolor[HTML]{38FFF8}82.866        & \cellcolor[HTML]{38FFF8}$\bigtriangleup$ & 2.004                               & $\times$                                & 83.452                              & $\bigtriangleup$                         \\
\multicolumn{1}{c|}{\multirow{-4}{*}{qft}}                          & 24  & 1488 & 1111 & \textgreater{}200                     & TO                                      & 2.497                               & $\times$                                & \textgreater{}200                     & TO                                       \\ \hline
\multicolumn{1}{c|}{}                                                 & 4   & 111  & 122  & 0.738                               & $\bigtriangleup$                         & 0.645                               & $\times$                                & \cellcolor[HTML]{38FFF8}0.573         & \cellcolor[HTML]{38FFF8}$\bigtriangleup$ \\
\multicolumn{1}{c|}{}                                                 & 8   & 319  & 334  & \cellcolor[HTML]{38FFF8}15.694        & \cellcolor[HTML]{38FFF8}$\bigtriangleup$ & 1.142                               & $\times$                                & 32.203                              & $\bigtriangleup$                         \\
\multicolumn{1}{c|}{\multirow{-3}{*}{qnn}}                          & 12  & 623  & 648  & \textgreater{}200                     & TO                                      & 4.624                               & $\times$                                & \textgreater{}200                     & TO                                       \\ \hline
\multicolumn{1}{c|}{}                                                 & 8   & 192  & 150  & \cellcolor[HTML]{38FFF8}0.326         & \cellcolor[HTML]{38FFF8}$\bigtriangleup$ & \cellcolor[HTML]{38FFF8}0.326         & \cellcolor[HTML]{38FFF8}$\bigtriangleup$ & 0.355                               & $\bigtriangleup$                         \\
\multicolumn{1}{c|}{}                                                 & 16  & 712  & 562  & \cellcolor[HTML]{38FFF8}0.33          & \cellcolor[HTML]{38FFF8}$\bigtriangleup$ & \cellcolor[HTML]{38FFF8}0.33          & \cellcolor[HTML]{38FFF8}$\bigtriangleup$ & 37.901                              & $\bigtriangleup$                         \\
\multicolumn{1}{c|}{}                                                 & 20  & 1092 & 861  & \cellcolor[HTML]{38FFF8}0.581         & \cellcolor[HTML]{38FFF8}$\bigtriangleup$ & 0.58                                & $\times$                                & \textgreater{}200                     & TO                                       \\
\multicolumn{1}{c|}{\multirow{-4}{*}{qpeexact}}                     & 24  & 1537 & 1173 & \textgreater{}200                     & TO                                      & 2.219                               & $\times$                                & \textgreater{}200                     & TO                                       \\ \hline
\multicolumn{1}{c|}{}                                                 & 8   & 192  & 150  & \cellcolor[HTML]{38FFF8}0.333         & \cellcolor[HTML]{38FFF8}$\bigtriangleup$ & \cellcolor[HTML]{38FFF8}0.333         & \cellcolor[HTML]{38FFF8}$\bigtriangleup$ & 0.608                               & $\bigtriangleup$                         \\
\multicolumn{1}{c|}{}                                                 & 16  & 712  & 562  & \cellcolor[HTML]{38FFF8}0.295         & \cellcolor[HTML]{38FFF8}$\bigtriangleup$ & \cellcolor[HTML]{38FFF8}0.295         & \cellcolor[HTML]{38FFF8}$\bigtriangleup$ & 32.205                              & $\bigtriangleup$                         \\
\multicolumn{1}{c|}{}                                                 & 20  & 1092 & 861  & \cellcolor[HTML]{38FFF8}0.623         & \cellcolor[HTML]{38FFF8}$\bigtriangleup$ & 0.62                                & $\times$                                & \textgreater{}200                     & TO                                       \\
\multicolumn{1}{c|}{\multirow{-4}{*}{qpeinexact}}                   & 24  & 1552 & 1175 & \cellcolor[HTML]{38FFF8}1.075         & \cellcolor[HTML]{38FFF8}$\bigtriangleup$ & 1.049                               & $\times$                                & \textgreater{}200                     & TO                                       \\ \hline
\multicolumn{1}{c|}{}                                                 & 3   & 141  & 113  & \cellcolor[HTML]{38FFF8}0.107         & \cellcolor[HTML]{38FFF8}$\bigcirc$       & \cellcolor[HTML]{38FFF8}0.107         & \cellcolor[HTML]{38FFF8}$\bigcirc$       & 0.635                               & $\bigcirc$                               \\
\multicolumn{1}{c|}{}                                                 & 4   & 357  & 294  & \cellcolor[HTML]{38FFF8}0.154         & \cellcolor[HTML]{38FFF8}$\bigtriangleup$ & \cellcolor[HTML]{38FFF8}0.154         & \cellcolor[HTML]{38FFF8}$\bigtriangleup$ & 0.677                               & $\bigtriangleup$                         \\
\multicolumn{1}{c|}{\multirow{-3}{*}{qwalk}}                        & 5   & 1305 & 1158 & \textgreater{}200                     & TO                                      & \textgreater{}200                     & TO                                      & \cellcolor[HTML]{38FFF8}18.207        & \cellcolor[HTML]{38FFF8}$\bigtriangleup$ \\ \hline
\multicolumn{1}{c|}{}                                                 & 4   & 66   & 44   & 0.388                               & $\bigtriangleup$                         & 0.26                               & $\times$                                & \cellcolor[HTML]{38FFF8}0.19         & \cellcolor[HTML]{38FFF8}$\bigtriangleup$ \\
\multicolumn{1}{c|}{}                                                 & 8   & 134  & 95   & 0.818                               & $\bigtriangleup$                               & 0.41                                & $\times$                                & \cellcolor[HTML]{38FFF8}0.216         & \cellcolor[HTML]{38FFF8}$\bigtriangleup$ \\
\multicolumn{1}{c|}{\multirow{-3}{*}{vqe}}                          & 16  & 270  & 211  & 2.878                               & $\bigcirc$                               & 2.335                               & $\times$                                & \cellcolor[HTML]{38FFF8}0.598         & \cellcolor[HTML]{38FFF8}$\bigcirc$       \\ \hline
\multicolumn{1}{c|}{}                                                 & 4   & 55   & 25   & 0.147         & $\bigtriangleup$ & 0.147         & $\bigtriangleup$ & \cellcolor[HTML]{38FFF8}0.116                               & \cellcolor[HTML]{38FFF8}$\bigtriangleup$                         \\
\multicolumn{1}{c|}{}                       & 8   & 127  & 61   & \cellcolor[HTML]{38FFF8}0.182         & \cellcolor[HTML]{38FFF8}$\bigtriangleup$ & 0.128                               & $\times$                                & 0.345                               & $\bigtriangleup$                         \\
\multicolumn{1}{c|}{\multirow{-3}{*}{wstate}}& 16   & 271  & 129   & TO         & $\times$ & 11.029                               & $\times$                                & TO                               & $\times$                         \\\hline
\end{tabular}

}
\caption{Comparison of Operation Modes.}
\label{tab:modeCmp}
\vspace{-6mm}
\end{table}

Our tool is publicly available under the MIT license. All experiments were conducted on an Azure Standard E8ads v5 server (8 vCPUs, 64 GiB memory).

The \toolname tool supports three operation modes. Assume that we are checking the equivalence of two circuits $C_1$ and $C_2$, all modes first construct the path-sum formula of $C_2(C_1)^\dagger$:
\begin{itemize}
    \item {\bfseries RR (Reduction-rules)} Applies only rule-based reductions during the path-sum formula construction. If the resulting formula differs from that of an identity matrix, the tool reports ``unknown.''
    \item {\bfseries WMC (Weighted Model Counting)} Constructs the path-sum formula without applying any reduction rules. It uses weighted model counting to determine if the final path-sum formula corresponds to the identity.
    \item {\bfseries Hybrid}  Initially applies reduction rules to simplify the formula. If the resulting formula is not syntactically identical to that of the identity, weighted model counting is employed to check for semantic equivalence.
\end{itemize}

\section{Experimental Evaluation}
To evaluate the performance of our proposed approach, we designed three sets of experiments.

\subsection{Comparison of Operation Modes}\label{sec:mode_compare}
In the first experiment, we compare the performance of the three operation modes (RR, WMC, and Hybrid) of \toolname. We use the benchmark suite from~\cite{mei2024equivalence}, which includes quantum algorithm circuits such as GHZ, QAOA, VQE, QNN, and Grover, among others. The original circuits are expressed using the gate set {Clifford+T, \Rx, \Ry, \Rz}. To facilitate comparison, we use \qiskit (version 1.4.2) to transpile these circuits to the gate set \{\Hg, \Ry, \Rz, \CNOT\}, and perform equivalence checking between the original and transpiled ones.


The details are presented in Table~\ref{tab:modeCmp}. In the table, $q$ denotes the number of qubits in the benchmark circuits, while $G_1$ and $G_2$ represent the total number of gates under the two different gate sets, respectively. For each operation mode, we record the runtime for equivalence checking and the corresponding result. For the \textbf{result}, the symbols $\bigcirc$ and $\bigtriangleup$ indicate that the two circuits are equivalent or equivalent up to a global phase, respectively, whereas ``TO'' denotes a timeout where the runtime exceeds 200 seconds. The symbol $\times$ indicates that the tool fails to verify equivalence. We highlight the best-performing mode in \tikz[baseline]{\fill[mygreen] (0,0) rectangle (0.4,0.2);} for each example.

Overall, the results align with our expectations: the \textbf{Hybrid} mode consistently outperforms the other two, particularly on more challenging benchmarks. Specifically, \textbf{RR} alone solves only 14 out of the 35 examples, while \textbf{WMC} solves 28, and \textbf{Hybrid} solves 29.
Examining the solving time for the largest instance within each algorithm category, we observe that \textbf{Hybrid} performs better for GHZ, GraphState, QFT, QNN, QPE-Exact, QPE-Inexact, QWalk, and WState constructions. In contrast, \textbf{WMC} outperforms Hybrid in only three categories: Grover, QAOA, and VQE benchmarks.

Upon further investigation, we found that the inefficiency in the three categories (Grover, QAOA, and VQE) where \textbf{WMC} outperforms \textbf{Hybrid} stems from a low \emph{hit rate}. Specifically, the percentage of instances where the path-sum formula matches the patterns required for applying a reduction rule is typically low, around 50\%. This reduced hit rate significantly limits the effectiveness of rule-based simplifications, making pure WMC more efficient in these cases.

\subsection{Comparison with Other RR and WMC Implementations}\label{sec:implementation_compare}

One might argue that the superior performance of the \textbf{Hybrid} mode could stem from an inefficient implementation of the \textbf{RR} or \textbf{WMC} modes within \toolname. To address this concern, we further compare \toolname against state-of-the-art tools: \feymann~\cite{feynman2010quantum}, which implements the path-sum reduction approach, and \quokka~\cite{mei2024equivalence}, which implements a WMC-based equivalence checking strategy, each using their own benchmark suites\footnote{We transpiled the Feynman benchmark to the basis gates [H, Y, Z, T, T$^\dagger$, CX] using Qiskit~1.2.4.  We encountered transpilation issues with Qiskit~1.4.2 on this set of benchmarks. 
All transpiled files are available on our GitHub repository.}. We want to again emphasize that although both \quokka and \textbf{WMC} reduced the circuit equivalence problem to weighted model counting, they use very different encodings. We still compare the two because they are already the closest.

\newcommand{\figurewidth}{\textwidth}
\begin{figure}[htb]
    \vspace{-0.3cm}
    \centering
    \begin{subfigure}[t]{\dimexpr 0.49\figurewidth\relax}
        \centering
        \begin{tikzpicture}[baseline=(current bounding box.north)]
          \begin{axis}[
            width=1.2\figurewidth,
            height=0.8\figurewidth,
            ymin=5e-3, ymax=200,
            symbolic x coords={ 
              tof\_3, tof\_4, mod5\_4, barenco\_tof\_3, barenco\_tof\_4, tof\_5, barenco\_tof\_5,
              vbe\_adder\_3, gf2\textasciicircum4\_mult, csla\_mux\_3, tof\_10, rc\_adder\_6,
              gf2\textasciicircum5\_mult, mod\_red\_21, barenco\_tof\_10, gf2\textasciicircum6\_mult,
              hwb6, csum\_mux\_9, qft\_4, mod\_mult\_55, gf2\textasciicircum7\_mult, qcla\_com\_7,
              gf2\textasciicircum8\_mult, ham15-low, gf2\textasciicircum9\_mult, qcla\_adder\_10,
              gf2\textasciicircum10\_mult, adder\_8, ham15-med, gf2\textasciicircum16\_mult,
              gf2\textasciicircum32\_mult, gf2\textasciicircum64\_mult
            },
            yticklabel style={font=\tiny},
            xtick=data,
            x tick label style={rotate=90,anchor=east,font=\tiny},
        enlarge x limits=0.02, 
            bar width=1.5pt,
            grid=both, grid style={dashed,gray!30},
            legend style={at={(0.5,1.02)},anchor=south,legend columns=2,font=\tiny}
          ]
            \addplot+[ybar,fill=blue!60,bar shift=-0.75pt,mark=none, area legend] coordinates { 
              (tof\_3,0.012) (tof\_4,0.018) (mod5\_4,0.022) (barenco\_tof\_3,0.026)
              (barenco\_tof\_4,0.027) (tof\_5,0.029) (barenco\_tof\_5,0.045) (vbe\_adder\_3,0.061)
              (gf2\textasciicircum4\_mult,0.085) (csla\_mux\_3,0.088) (tof\_10,0.102)
              (rc\_adder\_6,0.121) (gf2\textasciicircum5\_mult,0.137) (mod\_red\_21,0.138)
              (barenco\_tof\_10,0.188) (gf2\textasciicircum6\_mult,0.239) (hwb6,0.262)
              (csum\_mux\_9,0.291) (qft\_4,0.326) (mod\_mult\_55,0.352)
              (gf2\textasciicircum7\_mult,0.374) (qcla\_com\_7,0.41) (gf2\textasciicircum8\_mult,0.56)
              (ham15-low,0.603) (gf2\textasciicircum9\_mult,0.701) (qcla\_adder\_10,0.769)
              (gf2\textasciicircum10\_mult,0.947) (adder\_8,1.535) (ham15-med,1.703)
              (gf2\textasciicircum16\_mult,5.129) (gf2\textasciicircum32\_mult,20.757)
              (gf2\textasciicircum64\_mult,196.903)
            };
            \addplot+[ybar,fill=red!60,bar shift=0.75pt,mark=none, area legend] coordinates { 
              (tof\_3,0.026) (tof\_4,0.043) (mod5\_4,0.045) (barenco\_tof\_3,0.008)
              (barenco\_tof\_4,0.012) (tof\_5,0.044) (barenco\_tof\_5,0.021) (vbe\_adder\_3,0.089)
              (gf2\textasciicircum4\_mult,0.031) (csla\_mux\_3,0.079) (tof\_10,0.221)
              (rc\_adder\_6,0.371) (gf2\textasciicircum5\_mult,0.069) (mod\_red\_21,0.401)
              (barenco\_tof\_10,0.147) (gf2\textasciicircum6\_mult,0.153) (hwb6,0.722)
              (csum\_mux\_9,0.958) (qft\_4,0.058) (mod\_mult\_55,0.064)
              (gf2\textasciicircum7\_mult,0.309) (qcla\_com\_7,4.597) (gf2\textasciicircum8\_mult,0.657)
              (ham15-low,0.719) (gf2\textasciicircum9\_mult,1.018) (qcla\_adder\_10,18.814)
              (gf2\textasciicircum10\_mult,2.609) (adder\_8,15.425) (ham15-med,10.901)
              (gf2\textasciicircum16\_mult,37.899) (gf2\textasciicircum32\_mult,200)
              (gf2\textasciicircum64\_mult,200)
            };
            \legend{\textbf{RR}, \feymann}
          \end{axis}
        \end{tikzpicture}
        \caption{\textbf{RR} vs.\ \feymann}
        \label{fig:vsFeymann}
    \end{subfigure}
    \hfill
    \begin{subfigure}[t]{\dimexpr 0.49\figurewidth\relax}
        \centering
        \begin{tikzpicture}[baseline=(current bounding box.north)]
          \begin{axis}[
            width=1.2\figurewidth,
            height=0.8\figurewidth,
            ymin=0, ymax=200,
            symbolic x coords={
  ghz\_32, wstate\_4, vqe\_4, ghz\_64, vqe\_8, graphstate\_16, wstate\_8, qpeexact\_8,
  grover-noancilla\_4, qnn\_4, qaoa\_9, vqe\_16, qpeinexact\_8, ghz\_128, qwalk\_3,
  qwalk\_4, grover-noancilla\_5, graphstate\_32, qaoa\_7, qaoa\_11, qft\_16,
  graphstate\_64, qwalk\_5, qft\_20, grover-noancilla\_6, qnn\_8, qpeinexact\_16,
  qpeexact\_16, qft\_22, wstate\_16, wstate\_32
},
            yticklabel style={font=\tiny},
            xtick=data,
            x tick label style={rotate=90,anchor=east,font=\tiny},
            enlarge x limits=0.02,
            bar width=1.5pt,
            grid=both, grid style={dashed,gray!30},
            legend style={at={(0.5,1.02)},anchor=south,legend columns=2,font=\tiny}
          ]
            \addplot+[ybar,fill=blue!60,bar shift=-0.75pt,mark=none, area legend] coordinates { 
            (ghz\_32,0.08)             (wstate\_4,0.116)          (vqe\_4,0.19)              (ghz\_64,0.212)
            (qpeinexact\_8,0.608)      (qpeexact\_8,0.355)        (qnn\_4,0.573)             (qwalk\_3,0.635)
            (grover-noancilla\_4,0.446) (vqe\_8,0.216)             (ghz\_128,0.617)           (wstate\_8,0.345)
            (graphstate\_16,0.286)     (qaoa\_7,1.031)            (qwalk\_4,0.677)           (qaoa\_9,0.586)
            (grover-noancilla\_5,0.743) (qaoa\_11,1.274)           (vqe\_16,0.598)            (qft\_16,1.932)
            (graphstate\_32,0.778)     (qft\_20,20.706)           (graphstate\_64,2.427)     (qpeinexact\_16,32.295)
            (qnn\_8,32.203)            (qpeexact\_16,37.901)      (qft\_22,83.452)           (grover-noancilla\_6,28.579)
            (qwalk\_5,18.297)          (wstate\_16,200)           (wstate\_32,200)
            };
            \addplot+[ybar,fill=red!60,bar shift=0.75pt,mark=none, area legend] coordinates {
              (ghz\_32,0.017)   (wstate\_4,0.079)  (vqe\_4,0.095)   (ghz\_64,0.019)
      (qpeinexact\_8,1.82)  (qpeexact\_8,2.645) (qnn\_4,0.15)  (qwalk\_3,0.19)
      (grover-noancilla\_4,0.092) (vqe\_8,0.135) (ghz\_128,0.023) (wstate\_8,0.111)
      (graphstate\_16,0.032) (qaoa\_7,0.069) (qwalk\_4,1.399) (qaoa\_9,0.322)
      (grover-noancilla\_5,3.719) (qaoa\_11,1.142) (vqe\_16,0.258) (qft\_16,13.107)
      (graphstate\_32,0.05) (qft\_20,200) (graphstate\_64,0.113) (qpeinexact\_16,200)
      (qnn\_8,148.927) (qpeexact\_16,200) (qft\_22,200) (grover-noancilla\_6,164.546)
      (qwalk\_5,47.09) (wstate\_16,0.123) (wstate\_32,0.256)
            };
            \legend{\textbf{WMC}, \quokka}
          \end{axis}
        \end{tikzpicture}
        \caption{\textbf{WMC} vs.\ \quokka}
        \label{fig:vsQuokka}
    \end{subfigure}
    \vspace{-0.2cm}
    \caption{Comparisons between \toolname\ and other tools.}
    \vspace{-0.5cm}
    \label{fig:overallComparison}
\end{figure}
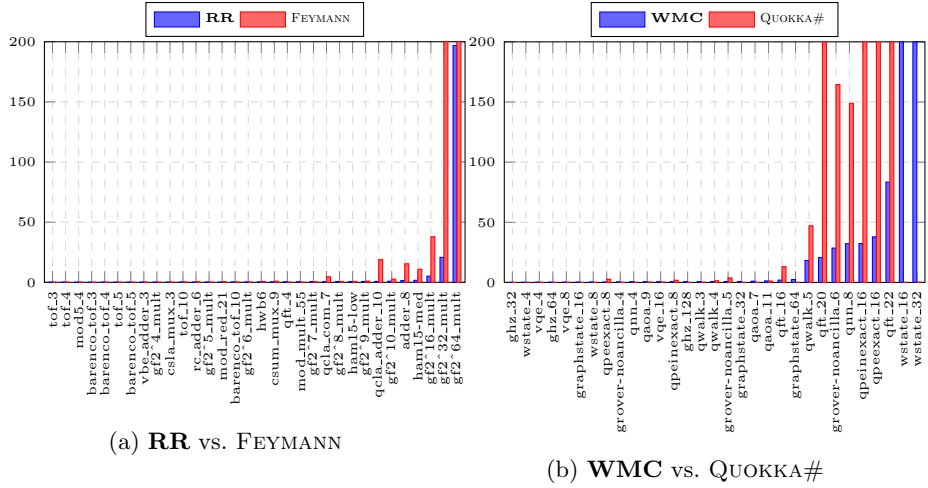

In Fig.~\ref{fig:vsFeymann}, we present a bar chart comparing the \textbf{RR} mode of \toolname with \feymann. The x-axis represents the benchmark example names, while the y-axis indicates the run time. A timeout of 80 seconds is enforced. The results demonstrate that the \textbf{RR} mode of \toolname consistently outperforms \feymann on more challenging examples, specifically, those requiring more than one second to solve.

We compare between the \textbf{WMC} mode of \toolname and \quokka in Fig.~\ref{fig:vsQuokka}. As before, the x-axis represents the benchmark example names, and the y-axis indicates the run time, with a timeout of 200 seconds enforced. The results show that the performance of the two tools is comparable. \quokka outperforms \toolname on four examples: graphstate 32, graphstate 64, wstate 16, and wstate 32. This outcome is somewhat expected. Although both tools use a weighted model counter as their backend engine, they adopt very different encodings of the path formula.  The performance of model counters is highly sensitive to the encoding; even for the same underlying concept, different encodings can lead to vastly different execution times. Therefore, it is unsurprising that \toolname performs better in some cases, while \quokka excels in others. Overall, we believe that it is fair to conclude that our implementation is comparable to the state-of-the-art implementation of the same algorithm.

\subsection{Comparison with Other Categories of Equivalence Checkers}\label{sec:tool_compare}

\begin{table}
\vspace{-6mm}
\centering
  \begin{subtable}{0.5\textwidth}
    \centering
    \scalebox{0.52}{

\begin{tabular}{cccc|cc|cc|cc|cc}
\hline
\multicolumn{4}{c|}{benchmarks}                                       & \multicolumn{2}{c|}{\toolname}                                           & \multicolumn{2}{c|}{\quokka}                                        & \multicolumn{2}{c|}{\qcec}                                          & \multicolumn{2}{c}{\pyzx}                                           \\ \hline
\multicolumn{1}{c|}{Name}                     & q   & G    & G2   & \multicolumn{1}{c}{time} & \multicolumn{1}{c|}{result} & \multicolumn{1}{c}{time} & \multicolumn{1}{c|}{result} & \multicolumn{1}{c}{time} & \multicolumn{1}{c|}{result} & \multicolumn{1}{c}{time} & \multicolumn{1}{c}{result} \\ \hline
\multicolumn{1}{c|}{}                                     & 32  & 34   & 33   & 0.027                                   & $\bigtriangleup$                         & \cellcolor[HTML]{FFFE65}0.017   & \cellcolor[HTML]{FFFE65}$\bigtriangleup$ & 0.03                         & $\bigtriangleup$                         & \cellcolor[HTML]{38FFF8}0.014 & \cellcolor[HTML]{38FFF8}$\bigtriangleup$ \\
\multicolumn{1}{c|}{}                                     & 64  & 66   & 65   & 0.084                                   & $\bigtriangleup$                         & \cellcolor[HTML]{38FFF8}0.019   & \cellcolor[HTML]{38FFF8}$\bigtriangleup$ & 0.035                        & $\bigtriangleup$                         & \cellcolor[HTML]{FFFE65}0.024 & \cellcolor[HTML]{FFFE65}$\bigtriangleup$ \\
\multicolumn{1}{c|}{\multirow{-3}{*}{ghz}}              & 128 & 130  & 129  & 0.306                                   & $\bigtriangleup$                         & \cellcolor[HTML]{38FFF8}0.023   & \cellcolor[HTML]{38FFF8}$\bigtriangleup$ & \cellcolor[HTML]{FFFE65}0.048 & \cellcolor[HTML]{FFFE65}$\bigtriangleup$ & 0.055                        & $\bigtriangleup$                         \\ \hline
\multicolumn{1}{c|}{}                                     & 16  & 160  & 108  & 0.162                                   & $\bigcirc$                               & \cellcolor[HTML]{FFFE65}0.032   & \cellcolor[HTML]{FFFE65}$\bigtriangleup$ & \cellcolor[HTML]{38FFF8}0.027 & \cellcolor[HTML]{38FFF8}$\bigtriangleup$ & 0.026                        & $\times$                                 \\
\multicolumn{1}{c|}{}                                     & 32  & 320  & 219  & 0.33                                    & $\bigcirc$                               & \cellcolor[HTML]{FFFE65}0.05    & \cellcolor[HTML]{FFFE65}$\bigtriangleup$ & \cellcolor[HTML]{38FFF8}0.031 & \cellcolor[HTML]{38FFF8}$\bigtriangleup$ & 0.045                        & $\times$                                 \\
\multicolumn{1}{c|}{\multirow{-3}{*}{graphstate}}       & 64  & 640  & 435  & 1.068                                   & $\bigtriangleup$                         & \cellcolor[HTML]{FFFE65}0.113   & \cellcolor[HTML]{FFFE65}$\bigtriangleup$ & \cellcolor[HTML]{38FFF8}0.044 & \cellcolor[HTML]{38FFF8}$\bigtriangleup$ & 0.143                        & $\times$                                 \\ \hline
\multicolumn{1}{c|}{}                                     & 4   & 190  & 160  & 0.213                                   & $\bigtriangleup$                         & \cellcolor[HTML]{FCFF2F}0.092   & \cellcolor[HTML]{FCFF2F}$\bigtriangleup$ & \cellcolor[HTML]{38FFF8}0.026 & \cellcolor[HTML]{38FFF8}$\bigtriangleup$ & 0.056                        & $\times$                                 \\
\multicolumn{1}{c|}{}                                     & 5   & 499  & 409  & \cellcolor[HTML]{FFFE65}0.194         & \cellcolor[HTML]{FFFE65}$\bigtriangleup$ & 3.719                        & $\bigtriangleup$                         & \cellcolor[HTML]{38FFF8}0.03  & \cellcolor[HTML]{38FFF8}$\bigtriangleup$ & 0.458                        & $\times$                                 \\
\multicolumn{1}{c|}{\multirow{-3}{*}{grover}} & 6   & 1568 & 1400 & \textgreater{}200                       & TO                                       & \cellcolor[HTML]{FFFE65}164.546 & \cellcolor[HTML]{FFFE65}$\bigtriangleup$ & \cellcolor[HTML]{38FFF8}0.076 & \cellcolor[HTML]{38FFF8}$\bigtriangleup$ & 4.734                        & $\times$                                 \\ \hline
\multicolumn{1}{c|}{}                                     & 7   & 133  & 170  & 1.023                                   & $\bigtriangleup$                         & \cellcolor[HTML]{FFFE65}0.069   & \cellcolor[HTML]{FFFE65}$\bigtriangleup$ & \cellcolor[HTML]{38FFF8}0.033 & \cellcolor[HTML]{38FFF8}$\bigcirc$       & 0.043                        & $\times$                                 \\
\multicolumn{1}{c|}{}                                     & 9   & 171  & 209  & 0.974                                   & $\bigtriangleup$                         & \cellcolor[HTML]{FFFE65}0.322   & \cellcolor[HTML]{FFFE65}$\bigtriangleup$ & \cellcolor[HTML]{38FFF8}0.056 & \cellcolor[HTML]{38FFF8}$\bigtriangleup$ & 0.122                        & $\times$                                 \\
\multicolumn{1}{c|}{\multirow{-3}{*}{qaoa}}             & 11  & 209  & 270  & 1.636                                   & $\bigtriangleup$                         & \cellcolor[HTML]{FFFE65}1.142   & \cellcolor[HTML]{FFFE65}$\bigtriangleup$ & \cellcolor[HTML]{38FFF8}0.065 & \cellcolor[HTML]{38FFF8}$\bigtriangleup$ & 0.141                        & $\times$                                 \\ \hline
\multicolumn{1}{c|}{}                                     & 16  & 672  & 543  & \cellcolor[HTML]{38FFF8}1.818         & \cellcolor[HTML]{38FFF8}$\bigtriangleup$ & \cellcolor[HTML]{FFFE65}13.107  & \cellcolor[HTML]{FFFE65}$\bigtriangleup$ & 0.712                        & $\times$                                 & 0.391                        & $\times$                                 \\
\multicolumn{1}{c|}{}                                     & 20  & 1040 & 825  & \cellcolor[HTML]{38FFF8}20.122        & \cellcolor[HTML]{38FFF8}$\bigtriangleup$ & \textgreater{}200                & TO                                       & 18.602                       & $\times$                                 & 0.693                        & $\times$                                 \\
\multicolumn{1}{c|}{}                                     & 22  & 1254 & 970  & \cellcolor[HTML]{38FFF8}82.966        & \cellcolor[HTML]{38FFF8}$\bigtriangleup$ & \textgreater{}200                & TO                                       & 157.183                      & $\times$                                 & 0.934                        & $\times$                                 \\
\multicolumn{1}{c|}{\multirow{-4}{*}{qft}}              & 24  & 1488 & 1111 & \textgreater{}200                       & TO                                       & \textgreater{}200                & TO                                       & \textgreater{}200              & TO                                       & 1.297                        & $\times$                                 \\ \hline
\multicolumn{1}{c|}{}                                     & 4   & 111  & 122  & 0.738                                   & $\bigtriangleup$                         & \cellcolor[HTML]{FFFE65}0.15    & \cellcolor[HTML]{FFFE65}$\bigtriangleup$ & \cellcolor[HTML]{38FFF8}0.031 & \cellcolor[HTML]{38FFF8}$\bigtriangleup$ & 0.047                        & $\times$                                 \\
\multicolumn{1}{c|}{}                                     & 8   & 319  & 334  & \cellcolor[HTML]{FFFE65}15.694        & \cellcolor[HTML]{FFFE65}$\bigtriangleup$ & 148.927                      & $\bigtriangleup$                         & \cellcolor[HTML]{38FFF8}0.041 & \cellcolor[HTML]{38FFF8}$\bigtriangleup$ & 0.172                        & $\times$                                 \\
\multicolumn{1}{c|}{\multirow{-3}{*}{qnn}}              & 12  & 623  & 648  & \textgreater{}200                       & TO                                       & \textgreater{}200                & TO                                       & \cellcolor[HTML]{38FFF8}0.173 & \cellcolor[HTML]{38FFF8}$\bigtriangleup$ & 1.412                        & $\times$                                 \\ \hline
\multicolumn{1}{c|}{}                                     & 8   & 192  & 150  & \cellcolor[HTML]{FFFE65}0.326         & \cellcolor[HTML]{FFFE65}$\bigtriangleup$ & 2.645                        & $\bigtriangleup$                         & \cellcolor[HTML]{38FFF8}0.042 & \cellcolor[HTML]{38FFF8}$\bigtriangleup$ & 0.126                        & $\times$                                 \\
\multicolumn{1}{c|}{}                                     & 16  & 712  & 562  & \cellcolor[HTML]{38FFF8}0.33          & \cellcolor[HTML]{38FFF8}$\bigtriangleup$ & \textgreater{}200                & TO                                       & \textgreater{}200              & TO                                       & 1.978                        & $\times$                                 \\
\multicolumn{1}{c|}{}                                     & 20  & 1092 & 861  & \cellcolor[HTML]{38FFF8}0.589         & \cellcolor[HTML]{38FFF8}$\bigtriangleup$ & \textgreater{}200                & TO                                       & \textgreater{}200              & TO                                       & 4.612                        & $\times$                                 \\
\multicolumn{1}{c|}{\multirow{-4}{*}{qpeexact}}         & 24  & 1537 & 1173 & \textgreater{}200                       & TO                                       & \textgreater{}200                & TO                                       & \textgreater{}200              & TO                                       & 7.145                        & $\times$                                 \\ \hline
\multicolumn{1}{c|}{}                                     & 8   & 192  & 150  & 0.333         & $\bigtriangleup$ & 1.82                         & $\bigtriangleup$                         & \cellcolor[HTML]{38FFF8}0.101 & \cellcolor[HTML]{38FFF8}$\bigtriangleup$ & \cellcolor[HTML]{FFFE65}0.167 & \cellcolor[HTML]{FFFE65}$\bigtriangleup$ \\
\multicolumn{1}{c|}{}                                     & 16  & 712  & 562  & \cellcolor[HTML]{38FFF8}0.295         & \cellcolor[HTML]{38FFF8}$\bigtriangleup$ & \textgreater{}200                & TO                                       & \textgreater{}200              & TO                                       & \cellcolor[HTML]{FFFE65}2.853 & \cellcolor[HTML]{FFFE65}$\bigtriangleup$ \\
\multicolumn{1}{c|}{}                                     & 20  & 1092 & 861  & \cellcolor[HTML]{38FFF8}0.629         & \cellcolor[HTML]{38FFF8}$\bigtriangleup$ & \textgreater{}200                & TO                                       & \textgreater{}200              & TO                                       & 8.937                        & $\times$                                 \\
\multicolumn{1}{c|}{\multirow{-4}{*}{qpeinexact}}       & 24  & 1552 & 1175 & \cellcolor[HTML]{38FFF8}1.075         & \cellcolor[HTML]{38FFF8}$\bigtriangleup$ & \textgreater{}200                & TO                                       & \textgreater{}200              & TO                                       & 18.478                       & $\times$                                 \\ \hline
\multicolumn{1}{c|}{}                                     & 3   & 141  & 113  & \cellcolor[HTML]{FFFE65}0.107         & \cellcolor[HTML]{FFFE65}$\bigcirc$       & 0.19                         & $\bigtriangleup$                         & \cellcolor[HTML]{38FFF8}0.096 & \cellcolor[HTML]{38FFF8}$\bigcirc$       & 0.179                        & $\times$                                 \\
\multicolumn{1}{c|}{}                                     & 4   & 357  & 294  & \cellcolor[HTML]{FFFE65}0.154         & \cellcolor[HTML]{FFFE65}$\bigtriangleup$ & 1.399                        & $\bigtriangleup$                         & \cellcolor[HTML]{38FFF8}0.121 & \cellcolor[HTML]{38FFF8}$\bigtriangleup$ & 1.256                        & $\times$                                 \\
\multicolumn{1}{c|}{\multirow{-3}{*}{qwalk}}   & 5   & 1305 & 1158 & \textgreater{}200                       & TO                                       & \cellcolor[HTML]{FFFE65}47.09   & \cellcolor[HTML]{FFFE65}$\bigtriangleup$ & \cellcolor[HTML]{38FFF8}0.225 & \cellcolor[HTML]{38FFF8}$\bigtriangleup$ & 5.025                        & $\times$                                 \\ \hline
\multicolumn{1}{c|}{}                                     & 4   & 66   & 44   & 0.388                                   & $\bigtriangleup$                         & \cellcolor[HTML]{38FFF8}0.095   & \cellcolor[HTML]{38FFF8}$\bigtriangleup$ & \cellcolor[HTML]{FFFE65}0.109 & \cellcolor[HTML]{FFFE65}$\bigtriangleup$ & 0.083                        & $\times$                                 \\
\multicolumn{1}{c|}{}                                     & 8   & 134  & 95   & 0.818                                   & $\bigtriangleup$                               & \cellcolor[HTML]{FFFE65}0.135   & \cellcolor[HTML]{FFFE65}$\bigtriangleup$ & \cellcolor[HTML]{38FFF8}0.125 & \cellcolor[HTML]{38FFF8}$\bigcirc$       & 0.116                        & $\times$                                 \\
\multicolumn{1}{c|}{\multirow{-3}{*}{vqe}}              & 16  & 270  & 211  & \cellcolor[HTML]{FFFE65}2.878         & \cellcolor[HTML]{FFFE65}$\bigcirc$       & \cellcolor[HTML]{38FFF8}0.258   & \cellcolor[HTML]{38FFF8}$\bigtriangleup$ & \textgreater{}200              & TO                                       & 0.276                        & $\times$                                 \\ \hline
\multicolumn{1}{c|}{}                                     & 4   & 55   & 25   & 0.147         & $\bigtriangleup$ & \cellcolor[HTML]{38FFF8}0.079   & \cellcolor[HTML]{38FFF8}$\bigtriangleup$ & \cellcolor[HTML]{FFFE65}0.102                        & \cellcolor[HTML]{FFFE65}$\bigtriangleup$                         & 0.072                        & $\times$                                 \\
\multicolumn{1}{c|}{}                                     & 8   & 127  & 61   & 0.182                                   & $\bigtriangleup$                         & \cellcolor[HTML]{FFFE65}0.111   & \cellcolor[HTML]{FFFE65}$\bigtriangleup$ & \cellcolor[HTML]{38FFF8}0.094 & \cellcolor[HTML]{38FFF8}$\bigtriangleup$ & 0.092                        & $\times$                                 \\
\multicolumn{1}{c|}{}                                     & 16  & 271  & 129  & \textgreater{}200                       & TO                                       & \cellcolor[HTML]{FFFE65}0.123   & \cellcolor[HTML]{FFFE65}$\bigtriangleup$ & \cellcolor[HTML]{38FFF8}0.11  & \cellcolor[HTML]{38FFF8}$\bigtriangleup$ & 0.083                        & $\times$                                 \\
\multicolumn{1}{c|}{\multirow{-4}{*}{wstate}}           & 32  & 559  & 261  & \textgreater{}200                       & TO                                       & \cellcolor[HTML]{FFFE65}0.256   & \cellcolor[HTML]{FFFE65}$\bigtriangleup$ & \cellcolor[HTML]{38FFF8}0.114 & \cellcolor[HTML]{38FFF8}$\bigtriangleup$ & 0.407                        & $\times$                                 \\ \hline
\end{tabular}
}
    \caption{Normal case.}
  \end{subtable}\hfill
  \begin{subtable}{0.5\textwidth}
    \centering
    \scalebox{0.52}{
\begin{tabular}{cccc|cc|cc|cc|cc}
\hline
\multicolumn{4}{c|}{benchmarks}                                            & \multicolumn{2}{c|}{\toolname}                   & \multicolumn{2}{c|}{WMC}                      & \multicolumn{2}{c|}{\quokka}         & \multicolumn{2}{c}{\qcec}      \\ \hline
\multicolumn{1}{c|}{Name}                    & q   & G    & G2   & time                               & result   & time                               & result   & time              & result          & time                               & result   \\ \hline
\multicolumn{1}{c|}{\multirow{3}{*}{ghz}}              & 32  & 34   & 34   & 0.04                               & $\times$ & 0.101                              & $\times$ & 0.02              & $\times$ & 0.061                              & $\times$ \\
\multicolumn{1}{c|}{}                                  & 64  & 66   & 66   & 0.095                              & $\times$ & 0.187                              & $\times$ & 0.021             & $\times$ & 0.098                              & $\times$ \\
\multicolumn{1}{c|}{}                                  & 128 & 130  & 131  & 0.352                              & $\times$ & 0.642                              & $\times$ & 0.023             & $\times$ & 0.172                              & $\times$ \\ \hline
\multicolumn{1}{c|}{\multirow{3}{*}{graphstate}}       & 16  & 160  & 110  & 0.16                               & $\times$ & 0.256                              & $\times$ & 0.032             & $\times$ & 0.076                              & $\times$ \\
\multicolumn{1}{c|}{}                                  & 32  & 320  & 222  & 0.399                              & $\times$ & 0.762                              & $\times$ & 0.051             & $\times$ & 0.403                              & $\times$ \\
\multicolumn{1}{c|}{}                                  & 64  & 640  & 440  & 1.223                              & $\times$ & 2.251                              & $\times$ & 0.131             & $\times$ & \textgreater{}200 & TO       \\ \hline
\multicolumn{1}{c|}{\multirow{3}{*}{grover}} & 4   & 190  & 162  & 0.478                              & $\times$ & 0.224                              & $\times$ & 0.1               & $\times$ & 0.047                              & $\times$ \\
\multicolumn{1}{c|}{}                                  & 5   & 499  & 414  & \textgreater{}200 & TO       & 0.969                              & $\times$ & 1.645             & $\times$ & 0.086                              & $\times$ \\
\multicolumn{1}{c|}{}                                  & 6   & 1568 & 1415 & 45.913                             & MO       & 36.179                             & $\times$ & 185.872           & $\times$ & 0.385                              & $\times$ \\ \hline
\multicolumn{1}{c|}{\multirow{3}{*}{qaoa}}             & 7   & 133  & 172  & 1.61                               & $\times$ & 0.328                              & $\times$ & 0.068             & $\times$ & 0.129                              & $\times$ \\
\multicolumn{1}{c|}{}                                  & 9   & 171  & 212  & 1.046                              & $\times$ & 0.6                                & $\times$ & 0.575             & $\times$ & 0.195                              & $\times$ \\
\multicolumn{1}{c|}{}                                  & 11  & 209  & 273  & 1.806                              & $\times$ & 1.039                              & $\times$ & 2.107             & $\times$ & 0.215                              & $\times$ \\ \hline
\multicolumn{1}{c|}{\multirow{4}{*}{qft}}              & 16  & 672  & 549  & 5.674                              & $\times$ & 4.675                              & $\times$ & 30.238            & $\times$ & \textgreater{}200 & TO       \\
\multicolumn{1}{c|}{}                                  & 20  & 1040 & 834  & 30.508                             & $\times$ & 32.808                             & $\times$ & \textgreater{}200 & TO         & \textgreater{}200 & TO       \\
\multicolumn{1}{c|}{}                                  & 22  & 1254 & 980  & \textgreater{}200 & TO       & \textgreater{}200 & TO       & \textgreater{}200 & TO         & \textgreater{}200 & TO       \\
\multicolumn{1}{c|}{}                                  & 24  & 1488 & 1123 & \textgreater{}200 & TO       & \textgreater{}200 & TO       & \textgreater{}200 & TO         & \textgreater{}200 & TO       \\ \hline
\multicolumn{1}{c|}{\multirow{3}{*}{qnn}}              & 4   & 111  & 124  & 0.838                              & $\times$ & 0.229                              & $\times$ & 0.113             & $\times$ & 0.073                              & $\times$ \\
\multicolumn{1}{c|}{}                                  & 8   & 319  & 338  & 13.421                             & $\times$ & 47.934                             & $\times$ & 168.344           & $\times$ & 2.131                              & $\times$ \\
\multicolumn{1}{c|}{}                                  & 12  & 623  & 655  & \textgreater{}200 & TO       & \textgreater{}200 & TO       & \textgreater{}200 & TO         & \textgreater{}200 & TO       \\ \hline
\multicolumn{1}{c|}{\multirow{4}{*}{qpeexact}}         & 8   & 192  & 152  & 45.512                             & $\times$ & 0.312                              & $\times$ & 1.662             & $\times$ & 0.089                              & $\times$ \\
\multicolumn{1}{c|}{}                                  & 16  & 712  & 568  & 0.993                              & $\times$ & 46.389                             & $\times$ & \textgreater{}200 & TO         & \textgreater{}200 & TO       \\
\multicolumn{1}{c|}{}                                  & 20  & 1092 & 870  & 9.718                              & $\times$ & \textgreater{}200 & TO       & \textgreater{}200 & TO         & \textgreater{}200 & TO       \\
\multicolumn{1}{c|}{}                                  & 24  & 1537 & 1185 & \textgreater{}200 & TO       & \textgreater{}200 & TO       & \textgreater{}200 & TO         & \textgreater{}200 & TO       \\ \hline
\multicolumn{1}{c|}{\multirow{4}{*}{qpeinexact}}       & 8   & 192  & 152  & 0.66                               & $\times$ & 0.199                              & $\times$ & 1.217             & $\times$ & 0.115                              & $\times$ \\
\multicolumn{1}{c|}{}                                  & 16  & 712  & 568  & 2.147                              & $\times$ & 49.435                             & $\times$ & \textgreater{}200 & TO         & \textgreater{}200 & TO       \\
\multicolumn{1}{c|}{}                                  & 20  & 1092 & 870  & 9.124                              & $\times$ & \textgreater{}200 & TO       & \textgreater{}200 & TO         & \textgreater{}200 & TO       \\
\multicolumn{1}{c|}{}                                  & 24  & 1552 & 1187 & 1.934                              & $\times$ & \textgreater{}200 & TO       & \textgreater{}200 & TO         & \textgreater{}200 & TO       \\ \hline
\multicolumn{1}{c|}{\multirow{3}{*}{qwalk}}  & 3   & 141  & 115  & 0.707                              & $\times$ & 0.231                              & $\times$ & 0.069             & $\times$ & 0.074                              & $\times$ \\
\multicolumn{1}{c|}{}                                  & 4   & 357  & 297  & 5.314                              & $\times$ & 0.638                              & $\times$ & 0.685             & $\times$ & 0.084                              & $\times$ \\
\multicolumn{1}{c|}{}                                  & 5   & 1305 & 1170 & \textgreater{}200 & TO       & 27.067                             & $\times$ & 24.648            & $\times$ & 0.397                              & $\times$ \\ \hline
\multicolumn{1}{c|}{\multirow{3}{*}{vqe}}              & 4   & 66   & 45   & 1.215                              & $\times$ & 0.094                              & $\times$ & 0.025             & $\times$ & 0.303                              & $\times$ \\
\multicolumn{1}{c|}{}                                  & 8   & 134  & 96   & 0.509                              & $\times$ & 0.228                              & $\times$ & 0.037             & $\times$ & 0.349                              & $\times$ \\
\multicolumn{1}{c|}{}                                  & 16  & 270  & 214  & 8.79                               & $\times$ & 0.919                              & $\times$ & 0.067             & $\times$ & \textgreater{}200 & TO       \\ \hline
\multicolumn{1}{c|}{\multirow{4}{*}{wstate}}           & 4   & 55   & 26   & 0.127                              & $\times$ & 0.077                              & $\times$ & 0.022             & $\times$ & 0.302                              & $\times$ \\
\multicolumn{1}{c|}{}                                  & 8   & 127  & 62   & 0.187                              & $\times$ & 0.331                              & $\times$ & 0.027             & $\times$ & 0.312                              & $\times$ \\
\multicolumn{1}{c|}{}                                  & 16  & 271  & 131  & \textgreater{}200 & TO       & \textgreater{}200 & TO       & 0.044             & $\times$ & 0.349                              & $\times$ \\
\multicolumn{1}{c|}{}                                  & 32  & 559  & 264  & \textgreater{}200 & TO       & \textgreater{}200 & TO       & 0.07              & $\times$ & 0.478                              & $\times$ \\ \hline
\end{tabular}
}
    \caption{Random rotation gate injected.}
  \end{subtable}
\caption{Comparison with State-of-the-art Tools}
\label{tab:compSOTA}
\end{table}

Finally, we evaluate \toolname against a broader selection of state-of-the-art quantum circuit equivalence checking tools, spanning various methodological categories. We selected the following tools for comparison. We highlight the best-performing tool in \tikz[baseline]{\fill[mygreen] (0,0) rectangle (0.4,0.2);} and second performing tool in \tikz[baseline]{\fill[myyellow] (0,0) rectangle (0.4,0.2);} for each example:

\begin{itemize} 
\item \quokka, a tool based on weighted model counting (WMC) 
\item \qcec~\cite{quetschlich2023mqtbench}, which combines three techniques ZX-calculus, QMDDs, and simulation, executed in parallel \item \feymann, which relies on circuit reduction rules 
\item \sliqec~\cite{wei2022accurate}, a tool based on binary decision diagrams (BDD) 
\item \pyzx~\cite{kissinger2020Pyzx}, which uses ZX-calculus \item \tdd~\cite{hong2022tensor}\cite{hong2023decision}\cite{hong2024equivalence}, which employs tensor decision diagrams (TDD) 
\end{itemize}

We found that both \feymann and \sliqec do not support rotation gates such as RX, RY, and RZ, and therefore are unable to handle the majority of our benchmark examples. \tdd frequently produces incorrect results, likely due to floating-point precision issues. Consequently, we only report detailed comparative results for \quokka, \qcec, and \pyzx in Table~\ref{tab:compSOTA}, using the same format as in Table~\ref{tab:modeCmp}. Here \toolname means the Hybrid model of the tool. The timeout for each tool is set to 200 seconds. 

From Table \ref{tab:compSOTA} (a), we observe that the three tools \toolname, \quokka, and \qcec are complementary. \toolname achieves the best results on the quantum Fourier transform and phase estimation algorithms, which involve many distinct rotation angles and require precise numerical reasoning about phase accumulation. \quokka performs better on the GHZ and variational quantum algorithms, while \qcec excels on graph-state preparation, Grover search, QAOA, quantum neural network, W-state, and quantum random walk benchmarks.

\toolname is particularly effective for circuits whose unitary transformations depend on continuous parameters such as rotation angles or phase shifts, where verifying equivalence requires reasoning about exact numerical values. In contrast, algorithms such as Grover’s search rely mainly on discrete structural symmetries and repetitive iterations, which are better handled by combinatorial reasoning as in \qcec. A unified framework combining the complementary strengths of all three tools would yield a state-of-the-art virtual equivalence checker.

In Table~\ref{tab:compSOTA}(b), we evaluate the robustness of the tools by injecting a random rotation gate into each circuit to simulate a phase error and check whether the tools can detect it. A $\times$ indicates that the injected error is successfully detected. For reference, we also include the WMC configuration of \toolname as a separate column. \pyzx is omitted because it fails to detect any of the injected errors.

\toolname detects the most errors, 28 out of 37 cases, demonstrating strong sensitivity to subtle numerical discrepancies. Its WMC mode detects the same number of buggy circuits but in different instances, suggesting complementary reasoning mechanisms within the framework. In comparison, \quokka detects 27 and \qcec 24 cases. These results highlight the superior capability of \toolname in identifying small yet meaningful deviations in circuit behavior.

\section{Conclusion}
We presented a hybrid verifier for quantum circuit equivalence that combines path-sum reductions with a new weighted model counting (WMC) encoding. The reductions aggressively simplify circuits when algebraic structure is present, and the WMC back end provides a complete semantic check up to global phase on the residual path-sum. We implemented the approach in \toolname, which supports reduction-only, WMC-only, and hybrid modes.
On standard benchmarks, the hybrid mode consistently outperforms either component alone and is competitive with state-of-the-art tools. It is especially effective on QFT and phase estimation, where correctness depends on many distinct rotation angles and precise phase accumulation. The robustness study with injected random rotations shows the highest error-detection rate among the compared tools. 
Overall, \toolname demonstrates that coupling symbolic reductions with complete model-based reasoning yields a practical and scalable path to quantum circuit equivalence checking.

\paragraph{Acknowledgments}
This work was supported by 
National Science and Technology Council, R.O.C., projects NSTC 114-2221-E-027-044 -MY2 and NSTC 114-2119-M-001-002-; Air Force Office of Scientific Research project FA2386-23-1-4107;
Academia Sinica Investigator Project Grant AS-IV-114-M07; Foxconn project;  French National
Research Agency (ANR) HQI project ANR-22-PNCQ0001;
Quantum Delta NL (QDNL) project.

\paragraph{Data Availability}
All data and code supporting the findings of this study are openly available in the repository \url{https://github.com/PhysicsQoo/pathsum}. The repository contains the implementation, example inputs, and artifacts required to reproduce the experimental results reported in this paper.

\bibliographystyle{plain}
\bibliography{Reference}

\end{document}